\documentclass[10pt,twocolumn,journal]{IEEEtran}
\usepackage{amsthm}
\usepackage{amssymb}
\usepackage{latexsym}
\usepackage{amsfonts}
\usepackage{amsbsy}
\usepackage{amsmath,amssymb}
\usepackage{times}
\usepackage{graphicx, hyperref}
\usepackage{enumerate}
\usepackage[usenames]{color}
\usepackage{epstopdf}
\usepackage{amsmath, url}
\usepackage{subfig}
\usepackage{bm}
\usepackage{booktabs}
\usepackage{cite}
\usepackage{color}
\usepackage{xcolor}
\usepackage{algpseudocode}
\usepackage{setspace}
\usepackage{bm}
\usepackage{tabularx}
\usepackage{amstext,epsfig,graphicx}
\usepackage{psfrag}
\usepackage{cite}
\usepackage{footnote}
\usepackage{siunitx}
\usepackage{mathrsfs}
\usepackage{booktabs}
\usepackage{multirow,bbm}
\usepackage{stfloats}
\usepackage{algorithm}
\usepackage{algorithmicx}
\usepackage{amsmath}
\usepackage{algorithmicx}
\usepackage{algorithm}
\usepackage{algpseudocode}

\algnewcommand\algorithmicinput{\textbf{INPUT: }}
\algnewcommand\Input{\item[\algorithmicinput]}
\algnewcommand\algorithmicoutput{\textbf{OUTPUT: }}
\algnewcommand\Output{\item[\algorithmicoutput]}

\usepackage{amsmath}
\DeclareMathOperator*{\argmax}{arg\,max}
\DeclareMathOperator*{\argmin}{arg\,min}

\allowdisplaybreaks[4] 
\usepackage{mathtools}
\mathtoolsset{showonlyrefs}
\begin{document}

%
\title{Fast Beam Alignment via Pure Exploration in Multi-armed Bandits}
%
%
%
\author{Yi Wei, Zixin Zhong, and Vincent Y. F. Tan, {\em Senior Member, IEEE}
\thanks{Y.~Wei is with the College of Information Science and Electronic Engineering, Zhejiang University and the Zhejiang Provincial Key Laboratory of Information Processing, Communication and Networking (IPCAN), Hangzhou 310027, China
(email: 21731133@zju.edu.cn). Z.~Zhong is with the Department of Computing Science, University of Alberta  (email:  zzhong10@ualberta.ca). V.~Y.~F.~Tan is with Department of Mathematics and the Department of Electrical and Computer Engineering, National University of Singapore (email:  vtan@nus.edu.sg).}
\thanks{This work is funded by a Singapore Ministry of Education Tier 1 grant (A-0009042-01-00),  the Fundamental Research Funds for the Central Universities 226-2022-00195, and the Defense Industrial Technology Development Program under Grant JCKY2020210B021. }
\thanks{This paper was presented in part at the 2022 International Symposium on Information Theory~\cite{wei2022}.}

}


\IEEEpeerreviewmaketitle
\maketitle

\newcommand\sFor[2]{ \For{#1}#2\EndFor} 
\newcommand\sIf[2]{ \If{#1}#2\EndIf}
\newtheorem{theorem}{\bf Theorem}
\newtheorem{lemma}{\bf Lemma}
\newtheorem{property}{\bf Property}

\begin{abstract}
The beam alignment (BA) problem consists in accurately aligning the transmitter and receiver beams to establish a reliable communication link in wireless communication systems.
 Existing BA methods search the entire beam space to identify the optimal transmit-receive beam pair. This incurs a significant latency when the number of  antennas is large. In this work, we develop a bandit-based fast BA algorithm to reduce BA latency for millimeter-wave (mmWave) communications. Our algorithm is named {\em Two-Phase Heteroscedastic Track-and-Stop} (2PHT\&S). We first formulate the BA problem as a pure exploration problem in multi-armed bandits in which the objective is  to minimize the required number of time steps given a certain fixed confidence level. By taking advantage of the correlation structure among beams that the information from nearby beams is similar and the heteroscedastic property that the variance of the reward of an arm (beam) is related to its mean, the proposed algorithm groups all beams into several beam sets such that the optimal beam set is first selected and the optimal beam is identified in this set after that. {Theoretical analysis and simulation results on synthetic and semi-practical   channel data  demonstrate the clear superiority of the proposed algorithm vis-\`a-vis other baseline competitors.}
\end{abstract}


\section{Introduction}
In millimeter-wave (mmWave) communications, the beams at both the transmitter and receiver are narrow directional.
{ The beam alignment (BA) problem consists in  ensuring  that the transmitter and receiver beams are accurately aligned to establish a reliable communication link \cite{8048526,8458146}.  An optimal transmitter (receiver) beam or a pair of transmitter and receiver beam is required to be selected to maximize the overall capacity,  throughput, SNR or diversity. }
{To achieve this goal, a na\"ive exhaustive search method scans all
the beam space and hence causes significant BA latency.
{ Specifically,  there are two fundamental challenges to implement BA:
(1) the  amount of additional channel state information
(CSI) corresponding to each beam (pair) is large; (2) the frequency of the channel measurement\cite{7947209,7990158} is high.}
These challenges become even more difficult to overcome when a large number of antennas are employed at both of the transmitter and receiver.
{ Moreover, the change of antenna orientation, the effect of transmitter/reciever mobility, and the dynamic nature of the wireless channel will result in previously found optimal beam pairs to be
suboptimal over time, which further exacerbates the latency problem \cite{7959169,8368998}.}}

Therefore, the design of  fast  BA algorithms is of paramount importance, and has stimulated intense research interest.
By utilizing the sparsity  of the mmWave channel, { Marzi {\em et al.} \cite{7390019}} incorporated compressed sensing techniques in the BA problem to reduce the beam measurement complexity.
{ Wang {\em et al.} \cite{5262295}} developed a fast-discovery multi-resolution beam search method, which first probes the wide beam  and continues to narrow beams
until the optimal arm is identified. However,  the beam resolution needs to be adjusted at each step.
Besides, various forms of side information, e.g., location information \cite{8852637}, out-of-band measurements \cite{8114345}, and dedicated short-range communication \cite{7786130}, have also been used to reduce the required number of probing beams.

Due to its inherent ability to balance between exploration and exploitation,
 multi-armed bandit (MAB) theory has been recently utilized in wireless communication field, such as channel access in cognitive network \cite{4723352}, channel allocation \cite{5535151} and  adaptive modulation/coding \cite{5518773}, etc., and   harnessed to address the BA problems; see \cite{6574205,8723104,8472783,8486279,8842625,8662770,9013578333}.
In a single-user MIMO system, the work \cite{6574205}   employed the canonical upper confidence bound (UCB) bandit algorithm in beam selection, where the instantaneous full CSI  is not required.
The work \cite{8723104} applied linear Thompson sampling (LTS) to address the beam selection problem.
{ For  mmWave vehicular systems, Sim {\em et al.}~\cite{8472783} developed an online learning algorithm
to address  the problem of beam selection with environment-awareness based on contextual bandit theory.} The proposed method explores different
beams over time while accounting for contextual information (i.e., vehicles' direction of arrival). Similarly, {Hashemi {\em et al.} \cite{8486279}}  proposed to  maximize the directivity gain (i.e., received energy) of the BA policy within a time period using the contextual information (i.e., the inherent correlation and unimodality properties), and formulate the BA problem as contextual bandits.
{ Wu {\em et al.} \cite{8842625} provided a method to quickly and accurately align beams in multi-path channels for a point-to-point mmwave
system, which
takes advantage of the correlation structure among beams such
that the information from nearby beams is extracted to identify
the optimal beam, instead of searching the entire beam space.}
In \cite{8662770}, { Va {\em et al.}} used a UCB-based framework to develop the online learning algorithm for beam pair selection
and refinement, where the algorithm first learns coarse beam directions in some predefined beam
codebook, and then fine-tunes the
identified directions to match the peak of the power angular spectrum at that position.
{ Hussain {\em et al.}}  \cite{9013578333} also designed a novel beam pair alignment scheme based on Bayesian multi-armed bandits, with the goal of maximizing the alignment probability and the data-communication
throughput.

In this paper, different from existing methods, we first formulate the BA problem as a {\em pure exploration} problem in the theory of MABs, {called {\em bandit BA problem}},
with the objective of minimizing the required time steps to find the optimal beam pair with a prescribed probability of success (confidence level).
{ Second, we derive a lower bound on the expected time complexity of {\em any} algorithm designed to solve the BA problem.} Third,
we present a bandit-based fast BA algorithm, which we name
{\em Two-Phase Heteroscedastic Track-and-Stop} (2PHT\&S). We do so
 by exploiting the correlation among beams and the heteroscedastic property that
the variance of the reward of an arm (beam) is linearly related to its mean. { Instead of searching over the entire
beam space, the proposed algorithm groups all beams into several beam sets
such that the optimal beam set is first selected and the optimal
beam is identified within this set during the second phase.
 We derive an upper bound on the time complexity of the proposed 2PHT\&S algorithm. Finally, extensive numerical results demonstrate that 2PHT\&S significantly reduces the required number of time steps to identify the optimal beam compared to existing algorithms in both simulated and semi-practical channel data scenarios.}

The remainder of this paper is organized as follows. The system model and problem
setup are presented in Section II. Section III proposes our bandit-based
fast BA algorithm and its time complexity upper bound. Simulation results are given in Section IV, and Section V concludes this paper.

\section{Problem Formulation}
In this section, we first describe the system model and propose our problem setup, after that we provide a generic  lower bound of the time complexity of the problem.
\subsection{System Model}
\begin{figure}[t]
\renewcommand{\captionfont}{\small}
\centering
\includegraphics[scale=.48]{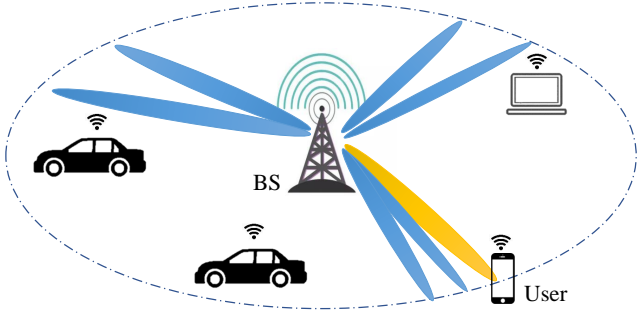}
\caption{A mmWave massive MISO system system.}
\label{system}
\normalsize
\end{figure}
As shown in Fig.~\ref{system}, we consider a massive mmWave MISO system, where a base station (BS) equipped with $N$ transmit antennas serves a single-antenna user. According to \cite{7967837}, each transmitting beam is selected from a predefined analog beamforming codebook $\mathcal{C}$ of size $K$, which can be defined as
\begin{equation}\label{codebook}
\mathcal{C}\triangleq \left\{\mathbf{f}_{k}=\mathbf{a}(-1+2 k / K) \mid k=0,1,2, \ldots, K-1\right\}
\end{equation}
with $\mathbf{a}(\cdot)$ denoting the array response vector. For a typical uniform linear array, it holds that
$
\mathbf{a}(x)=\frac{1}{\sqrt{N}}\big[1, e^{j \frac{2 \pi}{\lambda} d x}, e^{j \frac{2 \pi}{\lambda} 2 d x}, \ldots, e^{j \frac{2 \pi}{\lambda}(N-1) d x}\big]^{\mathrm{H}}\in \mathbb{C}^{N\times 1},
$
where $\lambda$ is the wavelength and $d$ is the antenna spacing.
{Note that if $K\le N$, the beams in the codebook $\mathcal{C}$ are linearly independent of one another.}
{ Based on channel measurement studies done in \cite{9689054,7857002}, the mmWave channel follows a multipath channel model, and the number of propagation paths is small. }

\begin{figure*}[t]
\renewcommand{\captionfont}{\small}
\centering
\includegraphics[scale=.6]{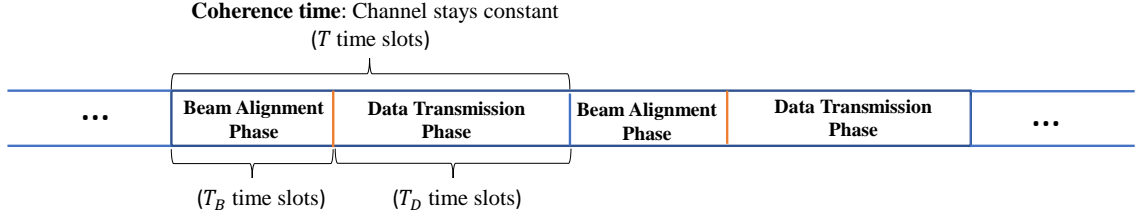}
\caption{Transmission scheme.}
\label{TP}
\normalsize
\end{figure*}
We consider a quasi-static channel, where the channel stays unchanged for a period covariance interval which consists of $T$ time slots.  { As shown in Fig. \ref{TP}, each covariance interval can be divided into two phases, the BA phase (which consists of $T_{\mathrm{B}}$ time slots) and the data transmission phase (which consists of $T_{\mathrm{D}} = T - T_{\mathrm{B}}$ time slots).}
With the selected beam $\mathbf{f}$, the effective achievable rate (EAR) of each covariance interval
\begin{equation}
R_{\text{eff}}:= \Big(1-\frac{T_{\mathrm{B}}}{T_{\mathrm{D}}}\Big)\log\Big(1 + \frac{ p|\mathbf{h}^{\mathrm{H}}\mathbf{f}|^2}{\sigma^2}\big)
\end{equation}
is widely accepted as a metric to measure the { throughput performance}, where $p$ and $\sigma^2$ represent the transmit power and noise variance respectively.
We  observe that the { throughput performance} increases  with the decreasing $T_{\mathrm{B}}/T_{\mathrm{D}}$ and increasing $|\mathbf{h}^{\mathrm{H}}\mathbf{f}^*|^2$. As such, the time allocated for selecting the optimal beam $\mathbf{f}^* = \argmax_{\mathbf{f}\in \mathcal{C}}|\mathbf{h}^{\mathrm{H}}\mathbf{f}|^2$ should be minimized for higher system throughput, which  is the motivation of this work.

\begin{figure}[t]
\renewcommand{\captionfont}{\small}
\centering
\includegraphics[scale=.5]{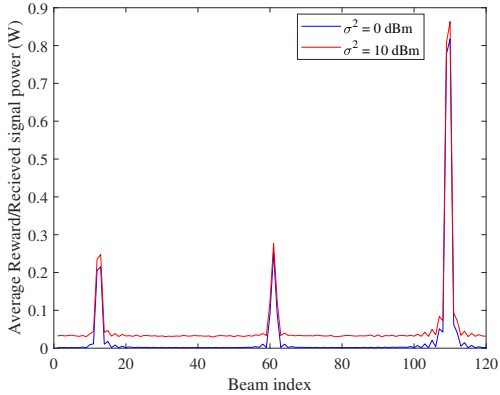}
\caption{ The average received power associated different beams over the codebook space in a three-path channel ($K = 120$, $N = 64$, $p = 26$ dBm and the power of line-of-sight (LoS) path is around 3 dB larger  than that of the non-line-of-sight (NLoS) path.)}
\label{AverRe}
\normalsize
\end{figure}
In the BA phase,
the BS selects one beam from the codebook $\mathcal{C}$ and transmits pilot signals  to the user with this beam. Without loss of generality, we set the pilot signal to be~$1$.
Then the received power can be expressed as
\begin{equation}\label{Reward}
R(\mathbf{f}_k) = |\sqrt{p}\mathbf{h}^{\mathrm{H}}\mathbf{f}_k + n |^2 = p|\mathbf{h}^{\mathrm{H}}\mathbf{f}_k|^2 + |n|^2 + 2\sqrt{p}\Re(\mathbf{h}^{\mathrm{H}}\mathbf{f}_kn^*),
\end{equation}
where $n^*$ represents the complex conjugate of the complex Gaussian noise  $n\sim \mathcal{CN}(0,\sigma^2)$.   {Note that the BS can observe the received power from the user through a feedback link, and this assumption is also adopted in other studies such as \cite{8842625,8662770}.}
{We can see that $R(\mathbf{f}_k)$ is a random variable which  is the sum of a Gamma random variable $|n|^2 \sim \Gamma(1, 1/\sigma^2)$ and a Gaussian random variable $p|\mathbf{h}^{\mathrm{H}}\mathbf{f}_k|^2 + 2\sqrt{p}\Re(\mathbf{h}^{\mathrm{H}}\mathbf{f}_kn^*) \sim \mathcal{N} ( p|\mathbf{h}^{\mathrm{H}}\mathbf{f}_k|^2, 2p |\mathbf{h}^{\mathrm{H}}\mathbf{f}_k|^2\sigma^2)$.
Fig. \ref{AverRe} shows the average received power associated with different beams over the codebook space when $K = 120$, $N = 64$, $p = 26$ dBm and the power of LoS path is around 3 dB larger than those of the NLoS paths. It can be observed  that all the received powers become larger as the increase of the noise variance, and the improved value is the difference between the noise variances.
Since the noise variance is much smaller than the transmit power, the variable $R(\mathbf{f}_k)$ is approximately a Gaussian random variable with mean $\mu_k = p|\mathbf{h}^{\mathrm{H}}\mathbf{f}_k|^2$ and variance $\sigma^2_k = 2p |\mathbf{h}^{\mathrm{H}}\mathbf{f}_k|^2\sigma^2 $, i.e.,
\begin{equation}\label{Reward2}
r_k = p|\mathbf{h}^{\mathrm{H}}\mathbf{f}_k|^2 +  2\sqrt{p}\Re(\mathbf{h}^{\mathrm{H}}\mathbf{f}_kn^*).
\end{equation}
The optimal beam is defined as the beam that has the largest value of $r_k$, i.e., ${k}^* = \argmax_k p|\mathbf{h}^{\mathrm{H}}\mathbf{f}_k|^2$.

 Let $J$ be the   \emph{correlation length} of arms, which is related to the number of beams $N$ and the size of the codebook $K$ as follows: $J = 2\lceil \frac{K}{N}\rceil-1$.} Then the received power using the beam $\mathbf{f}_k$ and the nearby beam $\mathbf{f}_i, |k-i|\le J/2$ are similar.  { A new $\frac{1}{J}$-lower resolution beam codebook $\mathcal{C}_{(J)}$ can be constructed by grouping the nearby beams in the codebook $\mathcal{C}$, i.e.,
\begin{equation}
\mathcal{C}_{(J)} \triangleq \Big\{\mathbf{b}_{g}= \frac{\sum_{k = J(g-1)+1}^{Jg}\mathbf{f}_k}{|\sum_{k = J(g-1)+1}^{Jg}\mathbf{f}_k|} : g=0,1, \ldots, G-1\Big\},
\end{equation}
where the beams $\mathbf{b}_g$'s are normalized (to unit $\ell_2$ norm) to mimic reality in a phased array. 
If we choose a beam $\mathbf{b}_g$ in the codebook $\mathcal{C}_{(J)}$, the received power can be expressed as
\begin{equation}
R_g = p|\mathbf{h}^{\mathrm{H}}\mathbf{b}_g|^2 +  2\sqrt{p}\Re(\mathbf{h}^{\mathrm{H}} \mathbf{b}_gn^*).
\end{equation}
Since the received powers associated with the nearby beams in the codebook $\mathcal{C}$ have means that are close to one another, $R_g$ has large mean if $r_k, k \in \{(g-1)J+1,\ldots,gJ  \}$ have large means.

\subsection{Problem Setup}{
To introduce the idea of bandit to address the BA problem, we consider a bandit model where the beam $\mathbf{f}_k$ is regarded as the base arm $k$, and the received power using this beam is regarded as the reward of the base arm $k$. The  means of the rewards when arms are pulled have the following property:
{ \begin{property}\label{Prop1}
Let ${\bm{\mu}} = (\mu_1,\ldots,\mu_K)$, and let $\mu_{(1)} \ge \mu_{(2)}\ge \mu_{(3)}\ge \ldots  \ge \mu_{(K)}$ be the sorted sequence of means.
Then we assume that these means have the following properties:
\begin{itemize}
 \item[1.] The means of the reward associated with arms $k$ and $i$, where $| i -k | \le J/2$, are close.
 \item[2.] {The rewards are sparse. In particular, there are only $LJ$ arms that have high rewards; the other $K-LJ$ arms have rewards that are close to zero.}
 \item[3.] The variance of the reward of an arm is related to its mean as follows:  ${\sigma}^2_k = 2{\mu}_k\sigma^2$.
 \end{itemize}
\end{property}}

{ As shown in Fig. \ref{AverRe}, Property 1.1 holds because the beam codebook is designed such that beams that are nearby have rewards that are close. Property 1.2 is due to the fact that there are only a few main paths in the considered channel model, which has been  corroborated by studies in  channel measurements  \cite{7109864}. Finally, Property 1.3 holds  because  the mean of reward associated with base arm $k$ follows a heteroscedastic Gaussian distribution $\mathcal{N}(\mu_k,2\mu_{k}\sigma^2)$ (see \eqref{Reward2}).} {Furthermore, in this work, we focus on a challenging   case with large noise, which satisfies that
\begin{equation}\label{sigma}
\sigma^2 > \max_{k \in [K]} \frac{(\mu_{(1)} - \mu_{(k)})^2}{4\mu_{(k)} - 2\mu_{(1)} -2\mu_{(k)}\ln\big(\frac{\mu_{(k)}}{\mu_{(1)}}\big) },
 \end{equation}
 and this condition is assumed to hold in our theoretical analyses. In the simulations, we also numerically check that this condition holds. }

{ Since a large number of base arms have mean rewards that are close to zero (Property 1.3), it is clearly a waste  of BA overhead to search for the optimal arm among all the base arms. To overcome this problem, we propose to group the base arms by utilizing Property 1.1 which says  that the  means of the reward associated with nearby arms are similar.
As such, we can  formulate the BA problem as a novel MAB problem, called the \em{bandit BA problem}.}}

In a \emph{bandit BA problem} with $K$ base arms, the base arm $k$ is associated with the beam $\mathbf{f}_k$.
We let $[K] = \{1,\ldots,K\}$ and { $\binom{[K]}{\le J}_{\mathrm{consec}} $} be the set of all non-empty consecutive tuples of length $\le J$ whose element belongs to the set $[K]$, e.g., if $J = 2$ and $ K = 6 $,
 \begin{align}
\binom{[6]}{\le 2}_{\mathrm{consec}}
 &= \Big\{\{1\},\{1,2\},\{2\},\{2,3\},\{3\}, \nonumber\\* &\qquad\{3,4\},\{4\},\{4,5\},\{5\},\{5,6\},\{6\}\Big\},
 \end{align}
 and we consider each tuple as a $(K,J)$-super arm. Moreover, the super arm associated with $\mathbf{b}_g\in \mathcal{C}_{(J)}$ is a subset of $\binom{[K]}{\le J}_{\mathrm{consec}}$, e.g., $\{\{1,2 \},\{3,4\},\{5,6 \} \} \subset \binom{[6]}{\le 2}_{\mathrm{consec}}
$.
{ At time step $t$, we choose an action (or a $(K,J)$-super arm) $A(t) \in \binom{[K]}{\le J}_{\mathrm{consec}} $, and observe the reward $R(t)$ which is related to the base arms included in $A(t)$, i.e.,
 \begin{equation}\label{Re1}
 R(t) = F\Big(\frac{\sum_{k\in {A(t)}}\mathbf{f}_k}{|\sum_{k\in {A(t)}}\mathbf{f}_k|}, p,\mathbf{h},n_t\Big),
 \end{equation}
where $ F(\mathbf{f},p, \mathbf{h},n)
 = p|\mathbf{h}^{\mathrm{H}}\mathbf{f}|^2 +  2\sqrt{p}\Re(\mathbf{h}^{\mathrm{H}}\mathbf{f}n^*)$, $n$ is a complex Gaussian random variable with mean zero and variance $\sigma^2$.
We can conclude that  for an arbitrary $(K,J)$-super arm $A \in \binom{[K]}{\le J}_{\mathrm{consec}}$, the empirical mean $R_A$ also follows a heteroscedastic Gaussian distribution, i.e., $R_A\sim \mathcal{N}(\mu_A, 2\mu_A\sigma^2)$, where $\mu_A = p|\mathbf{h}^{\mathrm{H}}\sum_{k\in {A}}\mathbf{f}_k|^2$.}

To identify the optimal base arm, an agent uses an algorithm
$\pi$ that decides which super arms to pull, the time $\tau^{\pi}$ to stop pulling,
and which base arm  $k^{\pi}$ to choose eventually.
An algorithm consists of a triple $\pi := \{ (\pi_t)_t,\tau^{\pi}, \psi^{\pi},J  \}$ in which:
\begin{itemize}
\item  a \underline{sampling rule} $\pi_t$  determining the $(K,J)$-super arm $A(t) \in\binom{[K]}{\le J }_{\mathrm{consec}}$ to pull at time step $t$ based on the observation history and the arm history $\{A(1),R(1),A(2),R(2),\ldots, A(t-1), R(t-1)\}$;
\item a \underline{stopping rule} leading to a stopping time $\tau^{\pi}$ satisfying $\mathbb{P}(\tau^{\pi}<\infty) = 1$;
\item a \underline{recommendation rule} $\psi^{\pi}$ that outputs a base arm $k^{\pi} \in [K]$.
\end{itemize}
We define the time complexity of $\pi$ as $\tau^{\pi}$. {In the fixed-confidence setting, a {\em maximum risk} $\delta\in (0,1)$ is fixed. We say an algorithm $\pi$ is {\em $(\delta,J)$-probably approximately correct} (PAC) if   $\mathbb{P}(k^{\pi} = k^*) \ge 1-\delta$.} The purpose is to design a $(\delta,J)$-PAC algorithm $\pi$ that minimizes the expected stopping time $\mathbb{E}[\tau^{\pi}]$.}

\subsection{Lower Bound on the Time Complexity}
Denote by $\nu$ a $K$-armed heteroscedastic Gaussian bandit instance, $\nu = \big\{\mathcal{N}(\mu^{\nu}_{1},2\mu^{\nu}_{1} \sigma^2),$ $ \ldots,\mathcal{N}(\mu^{\nu}_{K},2\mu^{\nu}_{K} \sigma^2) \big\}$.
Let $\mathcal{V}$ represent the set of $K$-armed heteroscedastic Gaussian bandit instances such that each bandit instance $\nu\in \mathcal{V}$ has a unique optimal arm:
for each $\nu\in\mathcal{V}$, there exists an arm $A^*(\nu)$ such that $\mu^{\nu}_{A^*(\nu)}>
\max\{\mu^{\nu}_{k}: k\neq A^*(\nu)  \}$. Define $\mathcal{W}_{K}=\big\{\mathbf{w} \in \mathbb{R}_{+}^{K}: \sum_{k=1}^K \vphantom{\Big[} w_{k} = 1\big\}$, $\operatorname{Alt}(\nu):=\left\{\mathrm{u} \in \mathcal{V}: A^{*}(\mathrm{u} ) \neq A^{*}(\nu)\right\}$, and $T_{k}(t)$ as the times the base arm $k$ is pulled until time step $t$.
Moreover, the KL divergence between two heteroscedastic Gaussian distributions, i.e., $\mathcal{N}(\mu_{i},2\mu_{i}\sigma^2)$ and $\mathcal{N}(\mu_{j},2\mu_{j}\sigma^2)$, can be calculated as
\begin{equation}
D_{\mathrm{HG}}(\mu_{i},\mu_{j})
= \frac{1}{2} \ln\Big(\frac{\mu_{j}}{\mu_{i}}\Big) + \frac{\mu_{i}}{2\mu_{j}} + \frac{(\mu_{j}-\mu_{i})^2}{4\mu_{j}\sigma^2} -\frac{1}{2}.
\end{equation}
According to Theorem 1 of \cite{Garivier2016}, we can obtain the general lower bound of this problem as follows.
\begin{theorem}\label{Theo1}
For any $(\delta,J)-$PAC algorithm
  where $\delta \in (0,1)$,  it holds that
\begin{equation}
  \mathbb{E}_{ \pi}[\tau] \geq c^{*}(\nu) \ln \left(\frac{1}{4\delta}\right),
  \end{equation}
  where
\begin{equation}\label{Theo2}
   c^{*}(\nu)^{-1}
     = \sup _{\mathbf{w} \in \mathcal{W}_{K}}\inf _{\mathrm{u} \in \operatorname{Alt}(\nu)}\Big(\sum_{k=1}^{K} w_{k}
    D_{\mathrm{HG}}(\mu^{\nu}_{k},\mu^{\mathrm{u}}_{k})
     \Big).
\end{equation}
  \end{theorem}

\section{Our Algorithm: 2PHT\&S}
{In this section, we first state the performance baseline of this work, i.e., the original Track-and-Stop (T\&S) algorithm~\cite{Garivier2016}, then we elaborate the design of the proposed 2PHT\&S. We then upper bound its time complexity in Theorem   \ref{theo2}.
\subsection{A Baseline Algorithm}
{ Note that the original Track-and-Stop (T\&S) \cite{Garivier2016} is the state-of-the-art  best arm identification algorithm with fixed confidence for exponential bandits, which can be modified  to solve the \emph{bandit BA problem} with low sample complexity, resulting in a {\em fast} BA solution. }

 Let $\nu$ represent a $K$-armed Gaussian bandit instance, $\nu = \big\{\mathcal{N}\big(\mu^{\nu}_{1},(\sigma^\nu_1)^2\big), \mathcal{N}\big( \mu^{\nu}_{2},(\sigma^\nu_2)^2\big), \ldots,$ $\mathcal{N}\big(\mu^{\nu}_{K},(\sigma^{\nu}_K)^2\big)  \big\}$.
Given the fixed  $\delta$,
   the time complexity of  T\&S, $\mathbb{E}\left[\tau^{\text{T\&S}}\right]$, satisfies  
\begin{equation}
\lim_{\delta \rightarrow 0} \frac{\mathbb{E}\left[\tau^{\text{T\&S}}\right]}{\log (1 / \delta)} = T^{*}(\nu),
\end{equation}
where
\begin{align}
&T^{*}(\nu)^{-1}\nonumber\\
&:=\!\sup _{\mathbf{w} \in \mathcal{W}_{K}} \!\inf _{\mathrm{u}\in \operatorname{Alt}(\nu)}\!\Big(\sum_{k=1}^{K} w_{k} D\left(\mathcal{N}\big(\mu^{\nu}_{k},(\sigma^\nu_k)^2\big), \mathcal{N}\big(\mu^{\mathrm{u}}_{k},(\sigma^{\mathrm{u}}_k)^2\big)\right)\Big).
\end{align}
and the KL-divergence $D\left(\mathcal{N}\big(\mu^{\nu}_{k},(\sigma^\nu_k)^2\big), \mathcal{N}\big(\mu^{\mathrm{u}}_{k},(\sigma^{\mathrm{u}}_k)^2\big)\right) =  \ln \big(\frac{\sigma^{\mathrm{u}}_k}{\sigma^{\nu}_ k}\big)+\frac{(\sigma^{\nu}_k)^{2}+\left(\mu^{\nu}_{k}-\mu^{\mathrm{u}}_k\right)^{2}}{2 (\sigma^{\mathrm{u}}_k)^{2}}-\frac{1}{2}$.

Furthermore, when we apply  T\&S to the  $K$-armed heteroscedastic Gaussian bandit instance $\nu \in \mathcal{V}$ under consideration, we have 
\begin{equation}\label{tl}
\limsup_{\delta \rightarrow 0} \frac{\mathbb{E}\left[\tau^{\text{T\&S}}\right]}{\log (1 / \delta)} \le T^*_{\mathrm{u}}(\nu),
\end{equation}
where
\begin{equation}
T^*_{\mathrm{u}}(\nu)
= \frac{8 \mu^{\nu}_{A^*(\nu)}\sigma^{2}}{\Delta_{\nu,\min }^{2}}+\sum_{k \neq A^*(\nu)} \frac{8 \mu^{\nu}_{A^*(\nu)}\sigma^{2}}{\Delta_{\nu,k}^{2}},
\end{equation}
and  the gaps are $\Delta_{\nu,k} := \mu^{\nu}_{A^*(\nu)} - \mu^{\nu}_{k}, k \neq A^*(\nu)$ and 
$\Delta_{\nu,\min } := \min_{k \neq A^*(\nu)} \mu^{\nu}_{A^*(\nu)} - \mu^{\nu}_{k}$.}
\subsection{2PHT\&S}
{

{ For our \emph{bandit BA problem},  we propose the  2PHT\&S  algorithm inspired by the original T\&S algorithm \cite{Garivier2016} to achieve a smaller time complexity by leveraging the structure of the BA problem. The main idea of the proposed algorithm is to exploit the prior knowledge which have not been considered by  existing algorithms, i.e., Property \ref{Prop1} and the fact that neighboring beams can be grouped and the corresponding ``grouped'' reward can be observed in one time step.}

\begin{figure*}[h]
\renewcommand{\captionfont}{\small}
\centering
\includegraphics[scale=.45]{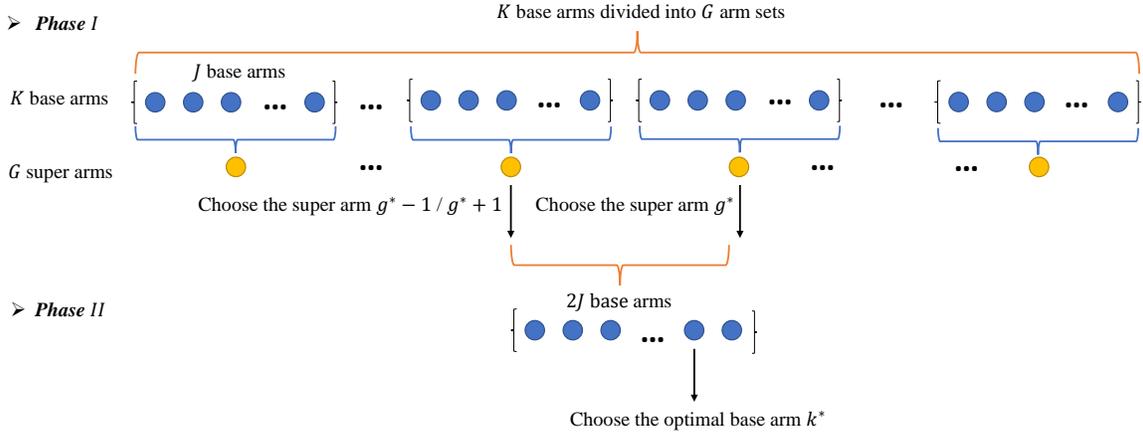}
\caption{Illustration of the proposed algorithms.}
\label{Algo}
\normalsize
\end{figure*}
As shown in  Fig.~\ref{Algo},  2PHT\&S consists of two phases. We describe {them} in the following.
\begin{itemize}
\item In Phase I, we first group $K$ base arms into $G$ arm sets and accordingly construct a set of super arms as
 \begin{equation}\label{eq12}
 \mathcal{G} = \{\mathcal{S}_g =  \{(g-1)J+1,\ldots, gJ\} | g\in [G] \}.
 \end{equation}
 According to \eqref{Re1}, the reward associated with super arm $\mathcal{S}_g$ can be written as $
 R_{g}(t) = F(\sum_{k\in {S_g}}\mathbf{f}_k, p,\mathbf{h},n_t)
$.  Second, we search for the optimal super arm, $g^* = \argmax_{\vphantom{\underline{g}}g\in[G]}\mathbb{E}[R_g(t)]$, with a high probability  of at least  $ 1-\delta_1$, where  $\delta_1$ denotes the { maximum allowable risk} in Phase I.
\item {In Phase II, we construct a base arm set (including two super arms) and select the optimal base arm from this base arm set.
     Considering the case shown in Fig. \ref{Case}, where the optimal base arm lies in the right edge of the super arm $\mathcal{S}_g$ and the second optimal base arm lies in the left edge of the super arm $\mathcal{S}_{g+2}$, the mean of the reward of the super arm $\mathcal{S}_{g+1}$ may be higher than that of the super arm $\mathcal{S}_{g}$.
     This observation results in a   higher difficulty of identifying the optimal super arm.
Therefore, to improve the selection accuracy,  we construct a base arm set, i.e., the combination of the optimal super arm and either one of its (left or right) neighbors that has the larger mean. Next,
we select the optimal base arm in this base arm set with a certain fixed confidence  probability of at least $1-\delta_2$.}
\end{itemize}
\begin{figure}[t]
\renewcommand{\captionfont}{\small}
\centering
\includegraphics[scale=.50]{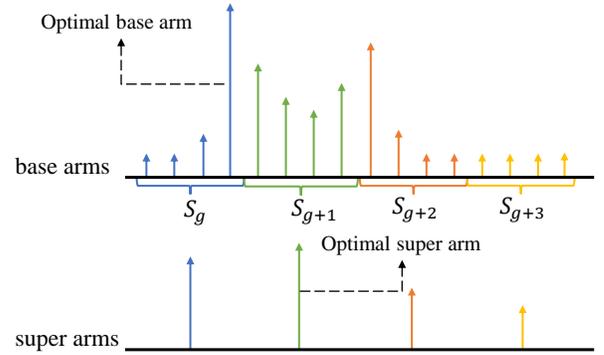}
\caption{Illustration of a case that the optimal base arm does not lie in the optimal super arm.}
\label{Case}
\normalsize
\end{figure}
Given the {fixed maximum risk}  $\delta$,   risks $\delta_1$ and $\delta_2$ should satisfy the constraint $\delta = \delta_1 + \delta_2$.\footnote{ Note that  since the two phases are independent, the probability of selecting the optimal base arm can be expressed as $(1-\delta_1)(1-\delta_2) = 1-\delta_1 -\delta_2 + \delta_1\delta_2$. Then, the error probability is $\delta_1+\delta_2 - \delta_1\delta_2$, which is less than $\delta_1+\delta_2$ since the two phases are independent. For simplicity, for the given fixed confidence $\delta$, we set $\delta_1 + \delta_2 = \delta$.} The proposed 2PHT\&S is detailed in Algorithm \ref{alg1}, where $\hat{\bm{\mu}}^{\mathrm{s}} = [\hat{\mu}^{\mathrm{s}}_1,\ldots, \hat{\mu}^{\mathrm{s}}_G]$ and  $\hat{\bm{\mu}}^{\mathrm{b}} = [\hat{\mu}^{\mathrm{b}}_{\mathcal{S}_f(1)},\ldots, \hat{\mu}^{\mathrm{b}}_{\mathcal{S}_f(2J)}]$ represent the empirical means of  super arms and base arms included in $\mathcal{S}_f$, respectively. In both of Phase I and Phase II of 2PHT\&S, an improved T\&S algorithm exploiting the heteroscedastic property, which is referred to as {\em Heteroscedastic Track-and-Stop} (HT\&S), is leveraged to search for the  optimal super arm and the optimal base arm, respectively.

\begin{algorithm}[t]
\setlength{\belowcaptionskip}{-0.0cm}
\setlength{\abovecaptionskip}{-0.0cm}
\caption{ 2PHT\&S} 
\label{alg1}
\hspace*{0.02in} {\bf Input:} {Maximum risk $\delta$} and a set of base arms $k\in \mathcal{K}$.\\
\hspace*{0.02in} {\bf Output:} Optimal base arm $k^*$, required number of time steps~$\tau$.
\begin{algorithmic}[1]
\State Choose $\delta_1$ and $\delta_2$ which satisfy  $ \delta_1 + \delta_2 = \delta$,

\noindent{\textbf{\emph{\# Phase I: search for the optimal super arm.}}}
\State Group the arms $k\in \mathcal{K}$ into $G$ arm sets and formulate the super arms $\mathcal{G} $ according to \eqref{eq12},
\State $[\tau_1,\hat{\bm{\mu}}^{\mathrm{s}}] = \text{HT\&S}(\delta_1,\mathcal{G})$ (defined in Algorithm \ref{alg2}),

\noindent{\textbf{\emph{\# Phase II: search for the optimal base arm.}}}
\State Select $g^*=\operatorname{argmax}_{g \in {\mathcal{G}}}\hat{{\mu}}^{\mathrm{s}}_{g}$,
 \If{$\hat{{\mu}}^{\mathrm{s}}_{g^*+1}(t)\ge \hat{{\mu}}^{\mathrm{s}}_{g^*-1}(t)$}
    \State $\mathcal{S}_f =  \mathcal{S}_{g^*}\cup \mathcal{S}_{g^*+1}$,
 \Else
   \State $\mathcal{S}_f =  \mathcal{S}_{g^*} \cup \mathcal{S}_{g^*-1}$,
 \EndIf
\State $[\tau_{2},\hat{\bm{\mu}}^{\mathrm{b}}] =\text{HT\&S}\big(\delta_2,\mathcal{S}_f\big)$ (defined in Algorithm \ref{alg2}),
\State \Return $k^* = \operatorname{argmax}_{k \in {\mathcal{S}}_f}\hat{{\mu}}^{\mathrm{b}}_{k}$ and $\tau = \tau_{1} + \tau_{2}$.
\end{algorithmic}
\end{algorithm}

 \begin{algorithm}[h]
\setlength{\belowcaptionskip}{-0.5cm}
\setlength{\abovecaptionskip}{-0.5cm}
\caption{HT\&S}
\label{alg2}
\hspace*{0.02in} {\bf Input:} Maximum risk $\delta$, super arm set $\mathcal{I}$.\\
\hspace*{0.02in} {\bf Output:} The required number of time steps $\tau$, the empirical mean of the reward $\hat{\bm{\mu}}^\nu$.
\begin{algorithmic}[1]
\For{$i \in \mathcal{I}$}
\State {Pull the super arm $i$, observe the reward, then update the $T_{i}(0)$ and $\hat{\mu}^\nu_{i}(0)$ according to \eqref{UP1}.}
\EndFor
\State Initialize $\hat{\mathbf{w}}^*(0)$ according to \eqref{upW} using $\hat{\bm{\mu}}^\nu(0)$, $t=1$.
\While{$Z(t)\le \beta(t,\delta,\alpha)$} (The definition of $Z(t)$ is given in \eqref{Z2}) \label{A2S}
    \If{$\operatorname{argmin}_{i \in {\mathcal{I}}} T_{i}(t-1) \leq (\sqrt{t} - \frac{I}{2} )^+$} \label{A2T1}
       \State $A(t) = \operatorname{argmin}_{i \in {\mathcal{I}}} T_{i}(t-1)$,
    \Else
       \State $A(t)=\operatorname{argmax}_{i \in {\mathcal{I}}}\left(t \hat{w}_{i}^{*}(t-1)-T_{i}(t-1)\right)$,
    \EndIf \label{A2T2}
    \State {Observe the reward $R(t)$, update $\hat{\mu}^\nu_{i}(t)$ and $T_{i}(t)$ according to \eqref{UP1}, update $Z(t)$ according to \eqref{Z2}, update $\hat{w}_{i}^{*}(t)$ according to \eqref{upW},}
    \State  $t = t+1$,
\EndWhile
\State $\tau_{\delta} = t$, $\hat{\bm{\mu}}^\nu = [\hat{\mu}^\nu_{i}(t),\ldots,\hat{\mu}^\nu_{I}(t)]$,
\State \Return $\tau = I+\tau_{\delta}$, $\hat{\bm{\mu}}^\nu$. \label{tau1}
\end{algorithmic}
\end{algorithm}
We now introduce the HT\&S algorithm. Define a set of $I \in \mathbb{N}$ super arms $\mathcal{I} = \{   \mathcal{S}_1,\ldots,\mathcal{S}_I\} $ where  the reward $R_{i}(t)$ of  super arm $\mathcal{S}_i$ is related to the beams $k\in \mathcal{S}_i$ as follows:
$
 R_{i}(t) = F(\sum_{k\in {S_i}}\mathbf{f}_k, p,\mathbf{h},n_t)$.
 Given $\mathcal{I}$ and the fixed $\delta$, HT\&S  outputs all the empirical means, $\hat{\bm{\mu}}^\nu = [\hat{\mu}^\nu_1,\ldots, \hat{\mu}^\nu_I]$, and the required number of time steps $\tau$. Let $A(t)$ represent the index of the super arm pulled in time step $t$, i.e.,
in time step $t$  the super arm $\mathcal{S}_{A(t)}$ is selected by  the sampling rule. We observe the following reward associated with the super arm  $ \mathcal{S}_{A(t)}$  is
\begin{equation}
R(t) = F\bigg(\sum_{k\in \mathcal{S}_{A(t)}}\mathbf{f}_{k}, p,\mathbf{h},n_t\bigg).
 \end{equation}
   Accordingly, we update   the number of times each arm is selected  up to time step $t$ and their empirical mean as
\begin{align}
 T_{i}(t) &=\sum_{a=1}^{t} \mathbbm{1}\left\{A(a) =i\right\}\quad\mbox{and}\\
 \hat{\mu}_{i}(t)& = \frac{1}{T_{i}(t)}\sum_{a=1}^{t} R(a)\mathbbm{1}\left\{A(a) = i\right\}.\label{UP1}
 \end{align}
The search is terminated when the stopping rule (delineated  in \eqref{STR1} below) is satisfied.
Define the heteroscedastic Gaussian bandit instance at time step $t$ as $\hat{\nu}_t = \big\{\mathcal{N}\big(\hat{\mu}^\nu_1(t),2\hat{\mu}^\nu_{1}(t) \sigma^2\big),\ldots,\mathcal{N}\big(\hat{\mu}^\nu_{I}(t),2\hat{\mu}^\nu_{I}(t) \sigma^2\big) \big\}$. The stopping rule and sampling rule, which exploit the heteroscedastic property, form the core of  HT\&S:

1) Stopping rule:  We choose the threshold as $\beta(t,{\delta},\alpha) = \ln(\alpha t/{\delta})$, and the stopping rule as
\begin{equation}\label{STR1}
\tau_{\delta} =\min\big\{t\in \mathbb{N}: Z(t)\ge \beta(t,{\delta},\alpha) \big\},
\end{equation}
where $Z(t)$ is defined in~\eqref{Z2} at the top of the next page,  and
\begin{figure*}
\begin{equation}\label{Z2}
\begin{aligned}
Z(t) &= \min_{i\neq A^*(\hat{\nu}_t)}\inf_{u\in \text{Alt}(\hat{\nu}_t)}
\Big\{T_{A^*(\hat{\nu}_t)}(t)D_{\text{HG}}(\hat{\mu}_{A^*(\hat{\nu}_t)}^{{\nu}}(t), \mu_{A^*(\hat{\nu}_t)}^u)
+  T_i(t)D_{\text{HG}}(\hat{\mu}_{i}^{\nu}(t), \mu_{i}^u)  \Big\}\\
& \overset{(a)}{=} \min_{i\neq A^*(\hat{\nu}_t) }\Big\{{T_{{A^*(\hat{\nu}_t)}}(t)}
      D_{\mathrm{HG}}\big(\hat{\mu}_{A^*(\hat{\nu}_t)}^{{\nu}}(t),q_i(t)\big)
      + {T_{i}(t)}D_{\mathrm{HG}}\big(\hat{\mu}_{i}^{{\nu}}(t),q_i(t)\big)
      \Big\}\\
     & = {T_{{A^*(\hat{\nu}_t)}}(t)}
      D_{\mathrm{HG}}\big(\hat{\mu}_{A^*(\hat{\nu}_t)}^{\nu}(t),q_{A'(\hat{\nu}_t)}(t)\big)
      + {T_{A'(\hat{\nu}_t)}(t)}D_{\mathrm{HG}}\big(\hat{\mu}_{A'(\hat{\nu}_t)}^{{\nu}}(t),
      q_{A'(\hat{\nu}_t)}(t)\big),
\end{aligned}
\end{equation}
\hrulefill
\end{figure*}
\begin{align}
q_{i}(t) &= \frac{{T_{A^*(\hat{\nu}_t)}(t)}  +  {T_{i}(t)}}{{T_{A^*(\hat{\nu}_t)}(t)} \hat{\mu}^\nu_{i}(t) + {T_{i}(t)} \hat{\mu}^\nu_{A^*(\hat{\nu}_t)}(t)}\hat{\mu}^\nu_{A^*(\hat{\nu}_t)}(t)\hat{\mu}^\nu_{i}(t), \\
A'(\hat{\nu}_t) &=  \argmax_{i\in [I]}\{\hat{\mu}_{i}^\nu(t): i\neq A^*(\hat{\nu}_t)  \}.\label{qfun}
\end{align}
Please refer to \eqref{TP32}--\eqref{TP33} in Appendix C for more details of the derivation of $(a)$ in \eqref{Z2}.

\begin{lemma}\label{lemma4}
Let $\nu$ be a heteroscedastic Gaussian bandit instance. Let ${\delta} \in (0,1)$ and $\alpha >1$. Using the above stopping rule in \eqref{STR1} with the threshold $\beta(t,{\delta},\alpha) = \ln(\frac{\alpha t}{{\delta}})$ ensures that $\mathbb{P}_{\mathbf{s}}\big(\tau_{{\delta}} < \infty, {A}_{\tau_{{\delta}}} \neq A^*(\nu)\big)\le {\delta}$.
\end{lemma}
\begin{proof}
Please refer to Appendix \ref{pr2}.
\end{proof}

2)  Sampling rule: The sampling rule can be summarized by $Q(t)$ defined in \eqref{eqn:Qt} on the top of the next page,
\begin{figure*}
 \begin{equation}
 Q(t) = \left\{\begin{array}{ll}
 \operatorname{argmin}_{i} T_{i}(t-1), & \text{ if  } \min_{i } T_{i}(t-1) \leq \sqrt{t}, \\
 \operatorname{argmax}_{i} t\hat{w}_{i}^*(t-1) - T_{i}(t-1), &\text{ otherwise}.
 \end{array}\right. \label{eqn:Qt}
 \end{equation}
 \hrulefill
 \end{figure*}
 where $\hat{\mathbf{w}}^*(t)  \in \mathcal{W}_I  $ is the set of parameters to be updated at time step $t$.
 Define
\begin{align}
& f_{Y,i,t}(x_i)  \nonumber\\ 
&:= \frac{1}{2}\ln\Big(\frac{\hat{\mu}^\nu_{i} (t)+ x_i \hat{\mu}^\nu_{A^*(\nu)}(t)}{(1 +  x_{i})\hat{\mu}^\nu_{i}(t)}\Big)
+ \frac{x_i}{2}\ln\Big(\frac{\hat{\mu}^\nu_{i}(t) + x_i \hat{\mu}^\nu_{A^*(\nu)}(t)}{(1 +  x_{i})\hat{\mu}^\nu_{A^*(\nu)}(t)}\Big)\nonumber\\ 
&\qquad+  \frac{ x_i (\hat{\mu}^\nu_{A^*(\nu)}(t)+\hat{\mu}^\nu_{i}(t))(\hat{\mu}^\nu_{A^*(\nu)}(t) - \hat{\mu}^\nu_{i}(t))^2 }{4\sigma^2 ( \hat{\mu}^\nu_{i}(t) + x_i \hat{\mu}^\nu_{A^*(\nu)}(t))},
\end{align} and $f_{X,i,t}(y)  \triangleq f^{-1}_{Y,i,t}(x_i)$.
{The value of  $\hat{w}^*_i(t)$ is set to be
 \begin{equation}\label{upW}
\hat{w}^*_{i}(t) =  \frac{f_{X,i,t}(y^*(t))}{\sum_{a=1}^I f_{X,a,t}(y^*(t)) },
\end{equation}
where  $y^*(t)$ is the unique solution of the equation
$\sum_{i\neq A^*(\bm{\mu})} {y f'_{X,i,t}(y) } -  f_{X,i,t}(y) = 1$. }
{Note that the choice of $\hat{\mathbf{w}}^*(t)$ is determined according to the optimization problem in~\eqref{Theo2}; see the technical details in Appendix \ref{PL6}.}\\

It is worth noting  that the stopping rule and sampling rule of  HT\&S are different from those of original T\&S, and both of them are derived from the lower bound given in Theorem~\ref{Theo1} based on the heteroscedastic property.

{
\subsection{Time Complexity Analysis of 2PHT\&S }
\begin{theorem}\label{theo2}
Let \begin{align}
    \mathrm{s}& = \big\{\mathcal{N}(\mu^{\mathrm{s}}_{1},2\mu^{\mathrm{s}}_{1} \sigma^2), \ldots,\mathcal{N}(\mu^{\mathrm{s}}_{G},2\mu^{\mathrm{s}}_{G} \sigma^2) \big\} \quad \mbox{and}  \\
    \mathrm{b} &= \big\{\mathcal{N}(\mu^{\mathrm{b}}_{\mathcal{S}_f(1)},2\mu^{\mathrm{b}}_{\mathcal{S}_f(1)} \sigma^2), \ldots,\mathcal{N}(\mu^{\mathrm{b}}_{\mathcal{S}_f(2J)},2\mu^{\mathrm{b}}_{\mathcal{S}_f(2J)} \sigma^2) \big\}
\end{align} be heteroscedastic Gaussian bandit instances studied in Phases~I and~II of 2PHT\&S, where $\mu^{\mathrm{s}}_{g} = p|\mathbf{h}^{\mathrm{H}}(\sum_{k\in \mathcal{S}_g}\mathbf{f}_k)|^2  $.
Using the proposed  stopping rule  and the sampling rule, we have
\begin{equation}\label{Lim}
\limsup\limits_{\delta\rightarrow 0}\frac{\mathbb{E}[\tau]}{\ln(1/\delta)}
\le c^*_{\mathrm{u}}(\mathrm{s}) + c^*_{\mathrm{u}}(\mathrm{b}),
\end{equation}
where
\begin{align}
c^*_{\mathrm{u}}(\mathrm{s}) &= \Big( \frac{\mu^{\mathrm{s}}_{A^*(\mathrm{s})}}{2\sum_{i=1}^G\mu^{\mathrm{s}}_{i}}\ln\big(\frac{2 \mu^{\mathrm{s}}_{A^*(\mathrm{s})}}{ \mu^{\mathrm{s}}_{A^*(\mathrm{s})} + \mu^{\mathrm{s}}_{A'(\mathrm{s})}} \big) \nonumber\\
&+ \frac{ \mu^{\mathrm{s}}_{A'(\mathrm{s})}}{2\sum_{i=1}^G\mu^{\mathrm{s}}_{i}}\ln\big(\frac{2  \mu^{\mathrm{s}}_{A'(\mathrm{s})}}{ \mu^{\mathrm{s}}_{A^*(\mathrm{s})} +  \mu^{\mathrm{s}}_{A'(\mathrm{s})}} \big) \nonumber \\
&+
\frac{(\mu^{\mathrm{s}}_{A^*(\mathrm{s})} - \mu^{\mathrm{s}}_{A'(\mathrm{s})})^2}{8\sigma^2\sum_{i=1}^G\mu^{\mathrm{s}}_{i} } - \frac{(\mu^{\mathrm{s}}_{A^*(\mathrm{s})}+\mu^{\mathrm{s}}_{A'(\mathrm{s})})}{2\sum_{i=1}^G\mu^{\mathrm{s}}_{i} }
\Big)^{-1},
\end{align} 
and
\begin{align}
c^*_{\mathrm{u}}(\mathrm{b}) &= \Big( \frac{\mu^{\mathrm{b}}_{A^*(\mathrm{b})}}{2\sum_{i \in \mathcal{S}_f}\mu^{\mathrm{b}}_{i}}\ln\big(\frac{2 \mu^{\mathrm{b}}_{A^*(\mathrm{b})}}{ \mu^{\mathrm{b}}_{A^*(\mathrm{b})} + \mu^{\mathrm{b}}_{A'(\mathrm{b})}} \big)\nonumber \\
&+ \frac{ \mu^{\mathrm{b}}_{A'(\mathrm{b})}}{2\sum_{i\in \mathcal{S}_f}\mu^{\mathrm{b}}_{i}}\ln\big(\frac{2  \mu^{\mathrm{b}}_{A'(\mathrm{b})}}{ \mu^{\mathrm{b}}_{A^*(\mathrm{b})} +  \mu^{\mathrm{b}}_{A'(\mathrm{b})}} \big) \nonumber \\
&+
\frac{(\mu^{\mathrm{b}}_{A^*(\mathrm{b})} - \mu^{\mathrm{b}}_{A'(\mathrm{b})})^2}{8\sigma^2\sum_{i\in \mathcal{S}_f}\mu^{\mathrm{b}}_{i} } - \frac{(\mu^{\mathrm{b}}_{A^*(\mathrm{b})}+\mu^{\mathrm{b}}_{A'(\mathrm{b})})}{2\sum_{i\in\mathcal{S}_f}\mu^{\mathrm{b}}_{i} }
\Big)^{-1}.
\end{align}
\end{theorem}

\begin{proof}
We first obtain the upper bound of the time complexity achieved by  HT\&S  in Lemma \ref{SC1}.
\begin{lemma}\label{SC1}
For a heteroscedastic Gaussian bandit instance $\nu$, Algorithm \ref{alg2} ensures that for any $\alpha \ge 1$ $t_0 >0$, $\epsilon > 0$ and $\delta>0$, there exists three constants $\Gamma_1 = \Gamma_1(\epsilon, t_0, I)$, $\Gamma_2 = \Gamma_2(\nu,\epsilon)$ and $\Gamma_3 = \Gamma_3(\nu,\epsilon)$ such that
\begin{align}
\mathbb{E}[\tau] & \le \Gamma_1 + \alpha c^*_{\mathrm{u}}(\nu)\left[\ln \left(\frac{ \alpha e c^*_{\mathrm{u}}(\nu)}{{\delta}}\right)+\ln \ln \left(\frac{\alpha c^*_{\mathrm{u}}(\nu)}{{\delta} }\right)\right] \nonumber \\ &\qquad+\sum_{T=1}^{\infty} \Gamma_2 T \exp \left(-\Gamma_3 T^{1 / 8}\right).
\end{align}
\end{lemma}
\begin{proof}
See the proof in Appendix \ref{PL9}.
\end{proof}
Using  HT\&S  in Phase I and Phase II,
the overall required number of time steps of Algorithm \ref{alg1}, i.e., $\tau$, can be expressed as
\begin{align}
\!\!\!\!\!\!\!\mathbb{E}[\tau]& = \mathbb{E}[\tau_1] + \mathbb{E}[\tau_2] \nonumber \\
 &\le \alpha_1 c^*_{\mathrm{u}}(\mathrm{s})\ln \left(\frac{1}{{\delta_1}}\right)
 \!+\! \alpha_2 c^*_{\mathrm{u}}(\mathrm{b})\ln \left(\frac{1}{{\delta_2}}\right)
 \!+\! M^{\mathrm{s}} \!+\! M^{\mathrm{b}},
 \label{Etau1}
 \end{align}
where
\begin{align}
M^{\mathrm{s}}& =  \Gamma_1^{\mathrm{s}} + \alpha_1 c^*_{\mathrm{u}}(\mathrm{s})\Big[\ln(\alpha_1e)+ \ln\ln\big(\frac{\alpha_1 c^*_{\mathrm{u}}(\mathrm{s}) }{\delta_1} \big)\Big]  \nonumber\\ 
&\qquad+ \sum_{T=1}^{\infty} \Gamma_2^{\mathrm{s}} T \exp \left(-\Gamma_3^{\mathrm{s}} T^{1 / 8}\right),\\
M^{\mathrm{b}} &=  \Gamma_1^{\mathrm{b}} + \alpha_2 c^*_{\mathrm{u}}(\mathrm{b})\Big[\ln(\alpha_2e)+ \ln\ln\big(\frac{\alpha_2 c^*_{\mathrm{u}}(\mathrm{b}) }{\delta_2} \big)\Big] \nonumber\\  
&\qquad+ \sum_{T=1}^{\infty} \Gamma_2^{\mathrm{b}} T \exp \left(-\Gamma_3^{\mathrm{b}} T^{1 / 8}\right). \label{CU}
\end{align}
For simplicity, we fix $\alpha_1 = \alpha_2 = 1$ and $\delta_1 = \delta_2 = \frac{1}{2}\delta$, then~\eqref{Etau1} can be rewritten as
\begin{align}
&\mathbb{E}[\tau] \le  c^*_{\mathrm{u}}(\mathrm{s})\ln \left(\frac{1}{2{\delta}}\right)
 + c^*_{\mathrm{u}}(\mathrm{b})\ln \left(\frac{1}{2{\delta}}\right)
 + M^{\mathrm{s}} + M^{\mathrm{b}}.
 \end{align}
When we normalize by $\ln (1/\delta)$ and let $\delta$ tend to zero, we obtain (24) as desired.

When $\delta$ tends to $0$, we obtain
\begin{equation}\label{Lim2}
\limsup\limits_{\delta\rightarrow 0}\frac{\mathbb{E}[\tau]}{\ln(1/\delta)}
\le c^*_{\mathrm{u}}(\mathrm{s}) + c^*_{\mathrm{u}}(\mathrm{b}),
\end{equation}
which concludes the proof.
\end{proof}

To summarize, we present in Table \ref{Tab3} the time complexities of the proposed 2PHT\&S  algorithm and its baseline, the original T\&S algorithm \cite{Garivier2016}, with fixed $\alpha = \alpha_1 = \alpha_2 = 1$. { In Table I, the terms $c^*_{\mathrm{u}}(\mathrm{s})$ and $c^*_{\mathrm{u}}(\mathrm{b})$ result from the time complexities of Phase I and Phase II, respectively. These constants respectively represent the difficulties of finding  the optimal super arm in Phase I and the optimal base arm in Phase II.} { There are two main differences between 2PHT\&S and T\&S. First, the proposed 2PHT\&S algorithm contains two phases. In the first phase, the BS searches for the optimal super arm (beam set), and  in the second phase, the BS searches for the optimal base arm (beam) among the beams in the optimal beam set and its neighbor with higher mean  reward. However, the vanilla T\&S algorithm directly searches for the optimal base arm (beam) among the whole beam space, which may be prohibitively large.
Second, the heteroscedastic property is explicitly considered in the design of the sampling rule and the stopping rule of the 2PHT\&S algorithm, but it is not exploited by the original T\&S algorithm.} { Furthermore, we would like to note that the computational complexity of T\&S algorithm scales linearly with  the number of beams $K$, while that of the proposed 2PHT\&S algorithm scales linearly with  the number of beam sets $G$ and the number of beams in each beam set $J$. Because $G$ and $J$ are usually much smaller than $K$, in general, the proposed 2PHT\&S algorithm can significantly reduce the computational complexity as compared to the original T\&S algorithm.}
{
\begin{table*}[t]
\caption{Time complexities}
\begin{center}
\begin{tabular}{|c|c|c|}
\hline
\textbf{Algorithm}& T\&S  & 2PHT\&S   \\
\hline
 Upper bound of $\lim\limits_{\delta\rightarrow 0}\frac{\mathbb{E}[\tau]}{\ln(1/\delta)}$     &   $T^*_{\mathrm{u}}(\nu)$ (see \eqref{tl}) &   $c^*_{\mathrm{u}}(\mathrm{s}) + c^*_{\mathrm{u}}(\mathrm{b})$  (see \eqref{CU}) \\
\hline
Lower bound of $\mathbb{E}[\tau]$ &  \multicolumn{2}{c|}{$c^{*}(\nu) \ln \left(\frac{1}{4\delta}\right)$ (see \eqref{Theo2})}\\
\hline
\end{tabular}
\label{Tab3}
\end{center}
\end{table*}
}}

{
\vspace{-0.0cm}
\section{Simulation Results}
{ We consider  a massive mmWave MISO system, where a BS equipped with $N=64$ transmit antennas serves a single-antenna user. The size of codebook is set as $K = 120$ and the correlation length is set to $J = 2\lceil\frac{K}{N} \rceil -1$.} Note that the
mmWave channel is sparse, and hence we set the maximum number of channel paths as $3$, which consists of one dominant LoS path and two NLoS paths.  In addition, according
to practical in-field measurements, NLoS paths suffer around  $1$ dB more path loss than the LoS path.
The noise variance is fixed to $\sigma^2 = -80$ dBm. The code to reproduce all the experiments can be found at this   Github link (\url{https://github.com/YiWei0129/Fast-beam-alignment}).

In the simulation, we compare the proposed 2PHT\&S  with three other bandit algorithms, i.e., the original T\&S  \cite{Garivier2016},  two-phase T\&S (2PT\&S), where the search process is also divided into two phases as the proposed algorithm and the original T\&S is used in each phase, and the proposed HT\&S. { Note that  2PT\&S  also exploits the correlation between the nearby beams (Property 1.1), but it does not take the heteroscedascity into account. The algorithms  we consider are all designed to search the optimal beam/base arm with high  probability  using as few time steps as possible.} { Moreover, we also investigate the 2PHT\&S algorithm with an overlapping scheme, i.e., in  Phase II, we construct the base arm set with $2J$ base arms where first $J$ base arms overlap with the last $J$ base arms of the super arm to the left of the selected super arm  in Phase I, and the last $J$ base arms overlap with the first $J$ base arms of the  super arm to the right of the selected super arm. We call this scheme 2PHT\&S (overlapping).}
{ Finally, we also consider the classical BA algorithms in wireless communications, i.e., exhaustive BA (EBA) algorithm and the hierarchical exhaustive BA (HEBA) algorithm as  performance baselines.\footnote{IIn the HEBA algorithm, the search phase is divided into two phases.  In each phase, the EBA algorithm is employed.}}

{ Additionally, note that in a typical mmWave communication setting, one may experience channel coherence times of 35 $\mu s$ when the system is deployed at the carrier frequency 28 GHz.
 In this setting, $T = 35$ $\mu s$, each time slot is 2.5 $ns$, and there are approximately 14000 time slots in each coherence time period.
 In practice, another link with different frequency can be used for feedback, thus no delay will  result.}

{ First, we fix $\delta = 0.1$, and $\delta_1 = \delta_2 = \frac{\delta}{2}$ for Phase I and Phase II.
We consider a scenario where the following widely-used channel model is considered, i.e.,
\begin{equation}
\mathbf{h}=\sqrt{\frac{N}{L}}\Big(\beta^{(1)} \mathbf{a}\big(\theta^{(1)}\big)+\sum_{l=2}^{L} \beta^{(l)} \mathbf{a}\big(\theta^{(l)}\big)\Big),
\end{equation}
 and the angle of departure (AoD) and path losses of three paths are set to $[0.7546\pi, 0.3489\pi, 0.6971\pi]$ and $[0, -3,-3]$ dB, respectively.  Fig.~\ref{Reward1} illustrates the means of the rewards associated with each base arm and super arm when $p = 0$~dBm. In this setting, we have $K = 120,N = 64,J = 2\lceil \frac{K}{N} \rceil - 1 = 3,G = \frac{K}{J}=40$. As can be seen, in this scenario, the AoDs of the three paths are different, such that the index of beams with non-zero means are separated by a certain  distance. Furthermore, the base arm 18 is the optimal arm, and the super arm $6$, which is related to the base arms $16$, $17$ and $18$, also has the largest mean of the reward.
Table \ref{tab2} presents the estimated average time complexities of the considered algorithms in this scenarios when the transmit signal-to-noise is chosen from the set $\{66, 70, 74,78 \}$ dB. It can be observed that  HT\&S and 2PHT\&S outperform  T\&S and 2PT\&S, since the former algorithms  explicitly exploit the heteroscedastic property in the considered {bandit BA problem}. Furthermore, the number of time steps of 2PHT\&S (2PT\&S) is much smaller than that of  HT\&S (T\&S). This improvement is  a consequence of dividing  the search process into two phases.}


\begin{figure}[t]
\renewcommand{\captionfont}{\small}
\centering
\includegraphics[scale=.35]{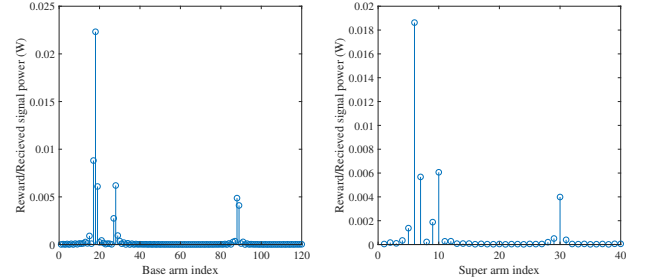}
\caption{ Mean of the reward of each base arm and super arm in Scenario 1 ($p = 0$ dBm). }
\label{Reward1}
\normalsize
\end{figure}

{
\begin{table*}[t]
\caption{ Average time complexities in Scenario 1 for $\delta = 0.1$, averaged over 100 experiments. }

\begin{center}
\begin{tabular}{|c|c|c|c|c|}
\hline
\textbf{SNR (dB)}& 66 & 70 & 74 & 78  \\
 \hline
   \textbf{EBA}& 4044.0 \tiny$  \pm  1089.7  $  & 1620.0 \tiny$  \pm 630.6$ & 744.0 \tiny$  \pm  308.8 $  & 348.0 \tiny $\pm 68.1$  \\
\hline
\textbf{T\&S} &3941.6 \tiny$  \pm  1177.8$ &1409.4 \tiny$  \pm  504.1 $ &  495.3 \tiny$  \pm  197.1 $   &  155.1\tiny$  \pm  49.6 $    \\
\hline
\textbf{HT\&S}&  1631.7 \tiny$  \pm   1028.6  $    & 577.8 \tiny$  \pm  349.0  $  &  213.2  \tiny$  \pm  110.3 $  & 128.6\tiny$\pm 23.9$ \\
 \hline
\textbf{HEBA}& 1562.6 \tiny$  \pm  356.4  $  & 630.4 \tiny$  \pm 136.3$ & 234.4 \tiny$  \pm  38.4 $  & 124.0 \tiny $\pm 16.9$  \\
\hline
\textbf{2PT\&S}&  1387.8 \tiny$  \pm 397.0  $   & 496.2 \tiny$  \pm  136.8 $  & 155.9 \tiny$  \pm  50.4$      &  58.9 \tiny$\pm 12.3$  \\
\hline
\textbf{2PHT\&S (overlapping)}& 452.1 \tiny$  \pm  322.3  $  & 270.8 \tiny$  \pm  124.3$ & \textbf{82.7} \tiny$  \pm  42.9 $  & 51.2 \tiny $\pm 4.3$  \\
\hline
\textbf{2PHT\&S}& \textbf{384.0} \tiny$  \pm  {302.9}  $  & \textbf{161.4} \tiny$  \pm  {82.4}$ &  {93.0} \tiny$  \pm  48.7 $  & \textbf{49.2} \tiny $\pm {2.5}$  \\
 \hline
\end{tabular}
\end{center}
\label{tab2}
\end{table*}
}

{ Second, we investigate a scenario in which the AoD of the LoS path and that of one NLoS path are close, i.e., the AoDs and path losses of three paths are $[0.3352\pi, 0.3521\pi, 0.8125\pi]$ and $[0, -3, -3]$ dB. In this scenario, since the AoDs of the main LoS path and one NLoS path are very close, the optimal base arm (base arm $91$) is the neighbor of the suboptimal base arm (base arm $90$), then the optimal super arm (super arm $31$) is also the neighbor of the suboptimal super arm (super arm $30$) and they have similar means of the rewards. Table~\ref{tab3} displays  the estimated average time complexities of the algorithms  we consider in the there scenarios when SNR  $\in \{70, 74, 78, 82\}$ dB. As compared to the results in Table \ref{tab2}, the time complexities of the  algorithms  under consideration  in Scenario 2 are higher.} Note that 
if the amplitudes of the largest path and the second largest path are   close, the time complexities of all algorithms will increase.

 \begin{table*}[h]
\caption{ Average time complexities in Scenario 2 for $\delta = 0.1$, averaged over 100 experiments. }

\begin{center}
\begin{tabular}{|c|c|c|c|c|}
\hline
\textbf{SNR (dB)}& 70 &74 & 78 & 82  \\
\hline
\textbf{EBA}&6518.4 \tiny$  \pm  1889.7  $  & 2373.6\tiny$  \pm 904.0$ &  912.0 \tiny$  \pm  309.5 $  & 441.6 \tiny $\pm 124.7$  \\
\hline
\textbf{T\&S} & 5981.1 \tiny$  \pm  1919.3 $   & 2119.5 \tiny$  \pm  1063.6 $ &  806.7 \tiny$  \pm  358.6 $   & 295.6 \tiny$  \pm  113.4 $  \\
 \hline
\textbf{HT\&S}&  3752.0 \tiny$  \pm   2029.4  $    & 1569.9 \tiny$  \pm 866.5  $  &  529.5  \tiny$  \pm  355.8 $  & 233.5 \tiny$\pm 109.3$ \\
\hline
\textbf{HEBA}& 1599.0 \tiny$  \pm  616.7  $  & 570.0 \tiny$  \pm  172.2$ &  249.6 \tiny$  \pm  65.8 $  & 134.6 \tiny $\pm 34.2$  \\
\hline
\textbf{2PT\&S}&  1437.6 \tiny$  \pm 427.7  $   & 509.0 \tiny$  \pm  203.2$  & 172.7 \tiny$  \pm  69.7 $      &  82.5\tiny$\pm 31.6$  \\
\hline
 \textbf{2PHT\&S}& \textbf{713.6} \tiny$  \pm  543.3 $  & \textbf{329.4} \tiny$  \pm  203.6$ &  \textbf{135.1} \tiny$  \pm  64.1$  & \textbf{70.8} \tiny $\pm 29.1$  \\
 \hline
\end{tabular}
\end{center}
\label{tab3}
\end{table*}

}}

{ Third, we consider a more realistic channel model from the 3GPP TR 38.901 standard, i.e., Scenario 3. The channel coefficients are generated by the QuaDRiGa simulator \cite{6758357}, which extends the popular Wireless World Initiative for New Radio (WINNER) channel model with new features to enhance its realism. In Fig. \ref{reward3}, the base arm $86$ is the
optimal arm to be chosen, and the super arm $29$, which is related to the base arms $85$, $86$ and $87$, has the largest means of the reward among super arms of size $3$.
Table \ref{tab4} presents the estimated average time complexities of the considered algorithms in Scenario 3 when SNR $\in \{78, 80, 82, 84\} $ dB. It can be observed that similar to the previous channel models, the proposed  2PHT\&S algorithm is clearly superior in terms of time complexity compared to  the other competing algorithms.}
\begin{figure}[t]
\renewcommand{\captionfont}{\small}
\centering
\includegraphics[scale=.35]{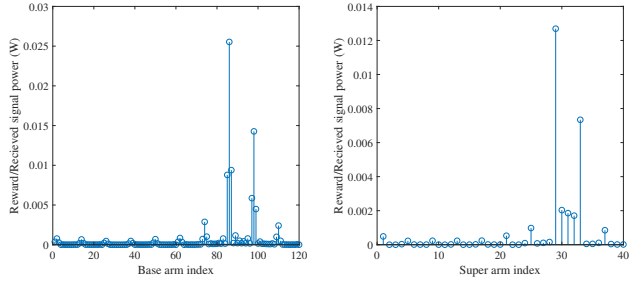}
\caption{ Mean of the reward of each base arm and super arm in Scenario 3 ($p = 0$ dBm). }
\label{reward3}
\normalsize
\end{figure}

 \begin{table*}[h]
\caption{ Average time complexities in Scenario 3 for $\delta = 0.1$, averaged over 100 experiments. }
\begin{center}
\begin{tabular}{|c|c|c|c|c|}
\hline
\textbf{SNR (dB)}& 78 &80 & 82 & 84  \\
 \hline
  \textbf{EBA}& 478.8 \tiny$  \pm  123.6  $  & 351.6 \tiny$  \pm  80.5$ &  289.2 \tiny$  \pm  61.7 $  & 244.8 \tiny $\pm 23.6$\\
\hline
\textbf{T\&S} & 357.2 \tiny$  \pm  147.7$ &  219.2 \tiny$  \pm  78.0 $ &  163.4 \tiny$  \pm  52.9 $   &  128.0\tiny$  \pm  23.9 $    \\
 \hline
   \textbf{HEBA}& 320.1 \tiny$  \pm  100.8  $  & 222.4 \tiny$  \pm  56.2$ &  161.6 \tiny$  \pm  43.8 $  & 126.0 \tiny $\pm 27.5$  \\
\hline
\textbf{2PT\&S}&  246.5 \tiny$  \pm 91.9  $   & 160.2 \tiny$  \pm  61.7 $  & 109.5 \tiny$  \pm  43.8 $      &  70.8\tiny$\pm 21.9$  \\
 \hline
\textbf{HT\&S}&  232.0 \tiny$  \pm   115.6  $    & 175.4 \tiny$  \pm  77.5  $  &  130.5  \tiny$  \pm  32.85 $  & 121.8\tiny$\pm 8.6$ \\
\hline
 \textbf{2PHT\&S}& \textbf{159.7} \tiny$  \pm  84.5  $  & \textbf{122.6} \tiny$  \pm  59.1$ &  \textbf{78.2} \tiny$  \pm  29.1 $  & \textbf{65.0} \tiny $\pm 18.9$  \\
 \hline
\end{tabular}
\end{center}
\label{tab4}
\end{table*}

{Finally, to validate the effectiveness of the proposed 2PHT\&S, we investigate a practical scenario in a city, i.e., Scenario 4, which is shown in Fig. \ref{Scena3}, where the BS and user are located at (573m, 622m, 41m) and (603m, 630m, 43m), respectively.   For this scenario, we generate the semi-practical channel data in the ray-tracing setups using the software \emph{Wireless InSite}.  Note that {\em  Wireless InSite} is a professional electromagnetic simulation tool that
models the physical characteristics of irregular
terrain and urban building features, performs the electromagnetic calculations, and then evaluates the
signal propagation characteristics.
As can be seen in Fig.~\ref{Reward3}, the base arm $119$ is the
optimal arm, and the super arm $40$, which is related to the base arm $118$, $119$ and $120$, has the largest means of the rewards among super arms of size $3$.
Table~\ref{Tab5} presents the estimated average time complexities of the considered algorithms in Scenario 4 when SNR $\in \{70, 74, 78, 82\}$ dB. It can be observed that similar to the previous simulated scenarios, the proposed  2PHT\&S outperforms  the other algorithms in this more practical scenario, which implies that the proposed algorithm is also effective in practice.}
   \begin{figure}[t]
\centering
\includegraphics[scale=.35]{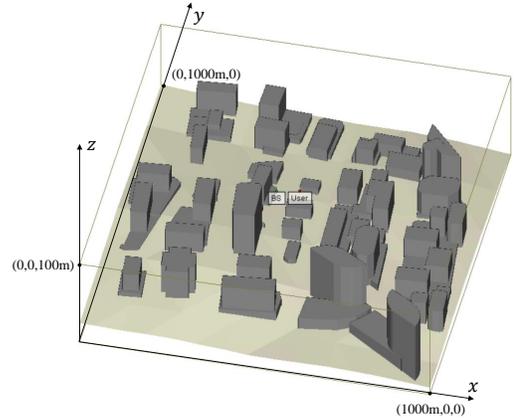}
\caption{Scenario 4:  A practical scenario in a city. }
\label{Scena3}
\normalsize
\end{figure}

\begin{figure}[t]
\centering
\includegraphics[scale=.35]{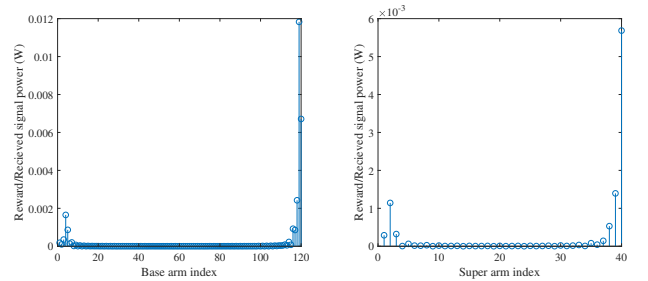}
\caption{Means of the rewards associated with each base arm and super arm in Scenario 4 ($p = 10$ dBm). }
\label{Reward3}
\normalsize
\end{figure}


\begin{table*}[h]
\caption{ Average time complexities in Scenario 4 for $\delta = 0.1$, averaged over 100 experiments. }

\begin{center}
\begin{tabular}{|c|c|c|c|c|}
\hline
\textbf{SNR (dB)}& 70 &74 & 78 & 82  \\
\hline
\textbf{T\&S} & 7944.5 \tiny$  \pm  2704.0$ &  2388.4 \tiny$  \pm  784.9 $ &  996.6\tiny$  \pm  357.8 $   &  372.7\tiny$  \pm  131.9 $    \\
\hline
  \textbf{EBA}& 6912 \tiny$  \pm  3235.4  $  & 2484 \tiny$  \pm  907.0$ &  1044.0 \tiny$  \pm  423.5
   $  & 576 \tiny $\pm 251.7$  \\
 \hline
 \textbf{HT\&S}&  2988.4 \tiny$  \pm  2307.2   $    & 1680.4 \tiny$  \pm  896.1  $  &  720.4 \tiny$  \pm  357.8 $  & 288.6\tiny$\pm 101.1$ \\
\hline
\textbf{HEBA}& 2948.8 \tiny$  \pm  666.4  $  & 1138.6 \tiny$  \pm  267.6$ &  409.0 \tiny$  \pm  68.2 $  & 202.4 \tiny $\pm 32.4$  \\
 \hline
\textbf{2PT\&S}&  2156.7 \tiny$  \pm 431.9  $   & 1164.4 \tiny$  \pm  313.6$  & 381.4 \tiny$  \pm  108.8 $      &  140.0\tiny$\pm 27.0$  \\
 \hline
 \textbf{2PHT\&S}& \textbf{561.9} \tiny$  \pm  390.4  $  & \textbf{312.0} \tiny$  \pm  157.1$ &  \textbf{161.0} \tiny$  \pm  109.2 $  & \textbf{48.0} \tiny $\pm 0$  \\
\hline
\end{tabular}
\end{center}
\label{Tab5}
\end{table*}

\section{Conclusion}
In this work, we developed a bandit-based fast BA algorithm 2PHT\&S to reduce BA latency for mmWave communications. By taking advantage of the correlation structure among beams that the information from nearby beams are similar and the heteroscedastic property that the variance of the reward of an arm (beam) is related to its mean, the proposed algorithm groups all beams into several beam sets such that the optimal beam set is first selected and the optimal beam is identified in this set. The proposed  2PHT\&S is shown to  theoretically and empirically perform much better than its baseline competitors.


\begin{appendices}
{
\section{Proof of Lemma \ref{lemma4}}\label{pr2}
\begin{proof}
Let $\mu_{i}^{\nu}$ be the mean of the reward of the arm $i$ in the heteroscedastic Gaussian bandit instance $\nu$.
{
First, for every arm $i_1$ and $i_2$ where $\mu_{i_1}^\nu< \mu_{i_2}^\nu$ and $\hat{\mu}^\nu_{i_1}(t) > \hat{\mu}^\nu_{i_2}(t)$, we define
\begin{align}
&Z_{i_1,i_2}(t)\nonumber\\
&= \inf\limits_{{\mu}^{u}_{i_1}< {{\mu}}^{u}_{i_2}}\Big\{T_{i_1}(t)D_{\mathrm{HG}}(\hat{\mu}^\nu_{i_1}(t), {\mu}^{u}_{i_1}) +
T_{i_2}(t)D_{\mathrm{HG}}(\hat{\mu}^\nu_{i_2}(t), {\mu}^{u}_{i_2})\Big\} \nonumber\\
&\le T_{i_1}(t)D_{\mathrm{HG}}(\hat{\mu}^\nu_{i_1}(t), {\mu}^\nu_{i_1}) +
T_{i_2}(t)D_{\mathrm{HG}}(\hat{\mu}^\nu_{i_2}(t), {\mu}^\nu_{i_2}).
\end{align}
}
 Then, according to the definition of $Z(t)$ given in \eqref{Z2}, we can rewrite $Z(t)$ as
\begin{align}
Z(t) &= \min_{i\neq A^*(\hat{\nu}_t)}\inf_{u\in \text{Alt}(\hat{\nu}_t)}
\Big\{T_{A^*(\hat{\nu}_t)}(t)D_{\text{HG}}(\hat{\mu}_{A^*(\hat{\nu}_t)}^{{\nu}}(t), \mu_{A^*(\hat{\nu}_t)}^u) \nonumber\\
&\qquad\qquad+  T_i(t)D_{\text{HG}}(\hat{\mu}_{i}^{\nu}(t), \mu_{i}^u)  \Big\} \\
& = \min\limits_{i \in [I]}Z_{A^*(\hat{\nu}_t), i}(t).
\end{align}
\begin{lemma}\label{lemma7}
For all $\psi \ge I+1$ and $t \ge 0$, it holds that
\begin{equation}
\mathbb{P}\bigg(\sum_{i=1}^I T_{i}(t) D_{\mathrm{HG}}\big(\hat{\mu}^\nu_{i}(t), \mu^\nu_{i}\big)\! \geq\! \psi \bigg)\le e^{I+1-\psi}\left( \frac{\lceil\psi \ln t  \rceil \psi}{I} \right)^I.
\end{equation}
\end{lemma}
\begin{proof}
See the proof in Appendix \ref{appB}.
\end{proof}
Now consider
\begin{align}
& \mathbb{P}\big(\tau_{{\delta}}<\infty, A_{\tau_{{\delta}}} \neq  A^{*}(\nu)\big) \nonumber\\&
\leq \mathbb{P}\Big(\exists\, i \in \mathcal{G} \backslash A^{*}(\nu), \exists\, t>0: \hat{\mu}^\nu_{i}(t)>\hat{\mu}^\nu_{ A^{*}(\hat{\nu}_t)},  \nonumber\\
&\qquad Z_{ A^{*}(\hat{\nu}_t),i}(t)>\beta(t, {\delta_1,\alpha_1})\Big) \\
& \leq \mathbb{P}\Big(\exists\,  t > 0: \exists\, i \in [I] \backslash A^{*}(\nu): \\ &\hspace{3em} T_{A^*(\hat{\nu}_t)}(t)D_{\mathrm{HG}}\big(\hat{\mu}^\nu_{A^*(\hat{\nu}_t)}(t), {\mu}^\nu_{A^*(\nu)}\big) + \nonumber\\
& \qquad T_{i}(t)D_{\mathrm{HG}}(\hat{\mu}^\nu_{i}(t), {\mu}^\nu_{i})>\beta(t, {\delta},\alpha)\Big)  \\
& \leq \mathbb{P}\Big(\exists\, t >0: \sum_{i\in [I]} T_{i}(t) D_{\mathrm{HG}}\big(\hat{\mu}^\nu_{i}(t), \mu^\nu_{i}\big) \geq \beta(t, {\delta},\alpha)\Big) \\
 &\overset{(a)}{\leq} \sum_{t=1}^{\infty} e^{I+1}\left(\frac{\beta(t, {\delta},\alpha)^{2} \ln (t)}{I}\right)^{I} e^{-\beta(t, {\delta},\alpha)}
\end{align}
where $(a)$ holds using the union bound and Lemma \ref{lemma7}.
Hence, with $\beta(t,{\delta},\alpha) = \ln(\alpha t/{\delta})$, by choosing an $\alpha$ satisfying
\begin{equation}
\sum_{t=1}^{\infty} \frac{e^{I+1}}{tI^I}\big[ \ln^2(\alpha t/{\delta})\ln(t)  \big]^I \le \alpha,
\end{equation}
we obtain $\mathbb{P}\big(\tau_{{\delta}}<\infty, A_{\tau_{{\delta}}} \neq A^{*}(\nu)\big)\le {\delta}$.
\end{proof}

\vspace{-0.4cm}
\section{Proof of Lemma \ref{lemma7}}\label{appB}
In this section, we prove Lemma~\ref{lemma7} assuming that Lemma~\ref{lemma8} holds, and then we complete the argument by proving Lemma~\ref{lemma8}.

\begin{lemma}\label{lemma8}
Fix an $I$-armed heteroscedastic Gaussian bandit $\nu$ and  let $1\le \bar{t}_i \le t$. Let $\eta> 0$. Define the event:
\begin{equation}\label{defiA}
E = \bigcap_{i=1}^I \{\bar{t}_i \le T_i(t) \le (1+\eta)\bar{t}_i  \}.
\end{equation}
For $\psi \ge (1+\eta)I $, it holds that
\begin{equation}
\mathbb{P}\Big(\mathbf{1}_E\sum_{i=1}^I T_{i}(t)D_{\mathrm{HG}}(\hat{\mu}_{i}^\nu(t), \mu_{i}^\nu )\ge \psi  \Big)
\le \big(\frac{\psi e}{I}\big)^I e^{-\psi/(1+\eta)}.
\end{equation}
\end{lemma}
\begin{proof}
See the proof in Appendix \ref{pr8}.
\end{proof}

\subsection{Proof of Lemma \ref{lemma7}}
\begin{proof}
Let $\psi \ge I+1$ and $\eta > 0$. Define $M = \lceil\ln(t)/\ln(1+\eta) \rceil$, and set
$\mathcal{M} = \{1,\ldots,M \}^I$ and events $A$ and $B_m$ as
\begin{equation}
\begin{aligned}
&A = \Big\{ \sum_{i=1}^I T_{i}(t) D_{\mathrm{HG}}\big(\hat{\mu}^\nu_{i}(t), \mu^\nu_{i}\big) \geq \psi    \Big\},\\
&B_m = \bigcap_{i=1}^I\Big\{  (1+\eta)^{m-1}\le T_i(t) \le (1+\eta)^{m}  \Big\}, \quad \forall\, m\in \mathcal{M}.
\end{aligned}
\end{equation}
Since $A = \cup_{m\in \mathcal{M}}(A\cap B_m)$, we have
$\mathbb{P}(A)\le \sum_{m\in \mathcal{M}}\mathbb{P}(A\cap B_m)$. Using Lemma \ref{lemma8},
we obtain for all $m\in \mathcal{M}$:
\begin{equation}
\mathbb{P}(A\cap B_m)\le \left(\frac{\psi e}{I}\right)^{I} e^{-\psi /(1+\eta)}
\end{equation}
Since $|\mathcal{M}| = M^I$, we have 
\begin{align}
\mathbb{P}(A)\le M^I \mathbb{P}(A\cap B_m)\le \left(\frac{M\psi e}{I}\right)^{I} e^{-\psi /(1+\eta)}.    
\end{align}
With the choice $\eta = 1/(\psi - 1)$ and the inequality $\ln(1+\eta)\ge 1-1/(1+\eta) = 1/\psi $, it holds that
\begin{equation}
\mathbb{P}(A) \le e^{-\psi}\Big(\frac{\psi\lceil\psi \ln (t)\rceil}{I}\Big)^{I} e^{I+1},
\end{equation}
which concludes the proof.
\end{proof}

\subsection{Proof of Lemma \ref{lemma8}}\label{pr8}
\begin{proof}
Define the event 
\begin{align}
  \mathcal{F} = \Big\{\mathbf{1}_E\sum_{i=1}^I T_{i}(t)D_{\mathrm{HG}}(\hat{\mu}^\nu_{i}(t),\mu^\nu_{i})\ge \psi \Big\}.  
\end{align}
 {Let $\zeta_i \in\left(\mathbb{R}^{+}\right)^{I}$, $T \in \mathbb{N}$ and $x_i(t)$ such that (i) if
there exists $0\le x\le \mu_i^\nu$ such that $TD_{\mathrm{HG}}(x,\mu_{i}^\nu) =\zeta_i$, then $x_i(T) = x$, (ii) else $x_i(T) = 0$. Since  $\mathrm{d} D_{\mathrm{HG}}(x,\mu_{i}^\nu)/\mathrm{d} x \le 0$ for $0\le x\le \mu_i^\nu$,  $D_{\mathrm{HG}}(x,\mu_{i}^\nu)$
is  a monotonically decreasing function of $x$ and furthermore $x_i(T)$ is a monotonically increasing function of $T$.} As a result, if $T_{i}(t)D_{\mathrm{HG}}(\hat{\mu}_{i}^\nu(t),\mu_{i}^\nu) \ge \zeta_i$, we have
\begin{equation}
\hat{\mu}_{i}^\nu(t)\overset{(a)}{\le} x_i(T_i(t))\overset{(b)}{\le} x_i((1+\eta)\bar{t}_i)
\end{equation}
where $(a)$ holds due to the  monotonicity of $x_i(T)$, and $(b)$ is due to the definition of event $E$ in \eqref{defiA}.
\begin{lemma}\label{Lemma9}
Fix an $I$-armed  heteroscedastic Gaussian bandit $\nu$ and let  $ 1\le \bar{t}_i \le t$. For any collection of  $0\le b_i \le \mu^\nu_{i}, i \in [I]$, we have
\begin{equation}
\mathbb{P}\Big(\cap_{i\in [I]}\big\{\hat{\mu}^\nu_{i}(t) \le b_i, \bar{t}_i\le T_i(t) \big\}  \Big)\le \prod_{i=1}^{I}e^{- \bar{t}_i D_{\mathrm{HG}}(b_i,\mu^\nu_{i})}.
\end{equation}
\end{lemma}
\begin{proof}
See Appendix \ref{pr9}.
\end{proof}

With this lemma, we can deduce that
\begin{align}
& \mathbb{P}\Big( \cap_{i\in [I]}\big\{ \mathbf{1}_E T_{i}(t)D_{\mathrm{HG}}(\hat{\mu}^\nu_{i}(t),\mu^\nu_{i})\ge \zeta_i \big\}  \Big) \nonumber\\
&\le \mathbb{P}\Big( \cap_{i\in [I]}\big\{ \hat{\mu}^\nu_{i}(t)\le x_i(T_{i}(t)),E  \big\} \Big)\nonumber\\
& \le \mathbb{P}\Big( \cap_{i\in [I]}\big\{ \hat{\mu}^\nu_{i}(t)\le x_i((1+\eta)\bar{t}_{i}),E  \big\} \Big)\nonumber\\
&\overset{(a)}{\le} \prod_{i=1}^{I}e^{- \bar{t}_i D_{\mathrm{HG}}(x_i((1+\eta)\bar{t}_{i}),\mu^\nu_{i})} 
\overset{(b)}{= }\prod_{i=1}^{I}e^{- \zeta_i /(1+\eta) }
\end{align}
where $(a)$ is obtained by Lemma \ref{Lemma9} with $b_i = x_i((1+\eta)\bar{t}_{i})$ and $(b)$ is because $\bar{t}_i D_{\mathrm{HG}}(x_i((1+\eta)\bar{t}_{i}),\mu_{i}) = \zeta_i /(1+\eta)$.
Lastly, by applying  \cite[Lemma~8]{Lipschitz}, we have
\begin{equation}
\mathbb{P}(\mathcal{F}) \le \Big(\frac{\psi e}{I(1+\eta)}\Big)^{I} e^{-\psi /(1+\eta)} \leq\Big(\frac{\psi e}{I}\Big)^{I} e^{-\psi /(1+\eta)},
\end{equation}
which concludes the proof.
\end{proof}

\subsection{Proof of Lemma \ref{Lemma9}}\label{pr9}
\begin{proof}
According to the Gibbs variational principle\cite{gibbs1902elementary}, one has the following dual characterization of the KL divergence: $D(P_1,P_2) = \sup_{f}(\mathbb{E}_{P_1}[f(X_1)]-\ln(\mathbb{E}_{P_2}[e^{f(X_2)}]))$, where the variables $X_1 \sim P_1$ and $X_2\sim P_2$. If we let $f(x) = \lambda x$, $\lambda \le 0$, for $0\le x \le \mu_{i}$, we have
\begin{equation}\label{KL2}
D_{\mathrm{HG}}(x,\mu^\nu_{i}) = \sup_{\lambda\le 0}\{\lambda x - \phi_{i}(\lambda_i)\},
\end{equation}
where $ \phi_{i}(\lambda_i) = \ln \mathbb{E}[e^{\lambda_i X_i}]  = \ln(\mu_i^\nu e^\lambda+1-\mu_i^\nu)$ and $X_i \sim \mathcal{N}(\mu_i^\nu,2\mu_i^\nu\sigma^2)$. Define the event $Q = Q_1 \cap Q_2$, where $Q_1 = \cap_{i\in [I]}\{ \bar{t}_i \le {T}_{i}(t)  \}$ and $Q_2 = \cap_{i\in [I]} \{\hat{\mu}^\nu_{i}(t)\le b_i\}$.

For all $i$, let $\lambda_i \le 0$ and define
$C(t) = \exp\big(\sum_{i\in [I]}\lambda_iS_{i}(t) - T_{i}(t)\phi_{i}(\lambda_i)  \big)  $, where $S_{i}(t)  =  T_{i}(t)\hat{\mu}^\nu_{i}(t)$.
For all $i$, we set $\lambda_i = \mathop{\argmax}_{\lambda\le 0}\{\lambda b_i - \phi_{i}(\lambda)  \}$ and $\lambda_ib_i -\phi_i(\lambda_i) = D_{\mathrm{HG}}(b_i, \mu^\nu_{i})$. Then, we have
\begin{align}
\mathbb{P}(Q) &=\mathbb{P}\left(\cap_{i\in[I]}\left\{S_{i}(t) \leq T_{i}(t) b_{i}, Q_{1}\right\}\right) \nonumber\\
 &\leq \mathbb{P}\Big(\sum_{i\in [I]} \lambda_{i} S_{i}(t) \geq \sum_{i\in [I]} \lambda_{i} T_{i}(t) b_{i}, Q_{1}\Big) \nonumber\\
& \leq \mathbb{P}\left(\mathbf{1}_{Q_{1}} e^{\sum_{i\in [I]} \lambda_{i} S_{i}(t)} \geq e^{\sum_{i\in [I]} \lambda_{i} T_{i}(t) b_{i}}\right) \nonumber\\
&=\mathbb{P}\left(\mathbf{1}_{Q_{1}} C(t) \geq e^{\sum_{i\in [I]} T_{i}(t)\left(\lambda_{i} b_{i}-\phi_{i}\left(\lambda_{i}\right)\right)}\right) \nonumber\\
&=\mathbb{P}\left(\mathbf{1}_{Q_{1}} C(t) \geq e^{\sum_{i\in [I]} T_{i}(t) D_{\mathrm{HG}}(b_i, \mu^\nu_{i})}\right)\nonumber\\
&\le \mathbb{P}\left(\mathbf{1}_{Q_{1}} C(t)>e^{\sum_{i\in [I]} \bar{t}_{i} D_{\mathrm{HG}}(b_i, \mu^\nu_{i})}\right).
\end{align}
Then, it holds that
\begin{align}
\mathbb{P}(Q) &\le \mathbb{P}\left(\mathbf{1}_{Q_{1}} C(t)>e^{\sum_{i\in [I]} \bar{t}_{i} D_{\mathrm{HG}}(b_i, \mu^\nu_{i})}\right) \nonumber\\ 
&\overset{(a)}{\le} \mathbb{E}[\mathbf{1}_{Q_1}C(t)]\exp\bigg(-\sum_{i\in [I]}\bar{t}_{i}
D_{\mathrm{HG}}(b_i,\mu^\nu_{i})\bigg) \nonumber\\
& \overset{(b)}{\le} \exp\bigg(-\sum_{i\in [I]}\bar{t}_{i}
D_{\mathrm{HG}}(b_i,\mu^\nu_{i})\bigg),
\end{align}
where $(a)$ is due to  Markov's inequality, {$(b)$ is due to the fact that $\mathbb{E}[\mathbf{1}_{Q_1}C(t)] \le \mathbb{E}[C(t)]$ and $ \mathbb{E}[C(t)] = 1$ is shown in the proof of  \cite[Lemma~7]{Lipschitz}}. This concludes the proof.
\end{proof}

\vspace{-0.4cm}
\section{}\label{PL6}
{
In this section, we first study the optimization problem \eqref{Theo2}, so as to better understand the choice of $\mathbf{w}^*$, then we provide an efficient method for computing $\mathbf{w}^*$.

By minimizing the KL divergence between the reward distributions of two bandits associated with arm $A^*(\nu)$ and arm $i\neq A^*(\nu)$,  \eqref{Theo2} can be transformed  into the following optimization problem
\begin{align}
& c^*(\nu)^{-1}  =   \sup _{\mathbf{w} \in \mathcal{W}_{I}}\inf _{\mathrm{u} \in \operatorname{Alt}(\nu)}\Big(\sum_{i=1}^{I} w_{k}
    D_{\mathrm{HG}}(\mu^{\nu}_{i},\mu^{\mathrm{u}}_{i})
     \Big) \nonumber\\
&= \sup_{\mathbf{w} \in \mathcal{W}_{I}}\inf _{ \substack{\mathrm{u} \in \mathcal{V}:  \\ i \neq A^*(\nu), \mu^{\mathrm{u}}_{i}\ge \mu^{\mathrm{u}}_{A^*(\nu)}}} \Big( \sum_{a\in \{A^*(\nu),i\}}
w_{a}
D_{\mathrm{HG}}(\mu^\nu_{a},\mu^{\mathrm{u}}_{a})\Big).\label{TP32}
\end{align}
It follows that
\begin{equation}
{\mathbf{w}}^* \!=\! \argmax_{\mathbf{w}\in \mathcal{W}_I}\inf _{\substack{\mathrm{u} \in \mathcal{V}:\\ i \neq A^*(\nu), \mu^{\mathrm{u}}_{i}\ge \mu^{\mathrm{u}}_{A^*(\nu)}}}\! \Big(\sum_{a\in \{A^*(\nu),i\}}
w_{a}
D_{\mathrm{HG}}(\mu^\nu_{a},\mu^{\mathrm{u}}_{a})\Big).
\end{equation}
 By using  Lagrange multipliers, we can solve the optimization problem
\begin{align}
 \min_{\mu^{\mathrm{u}}_{A^*(\nu)},\mu^{\mathrm{u}}_{i}} & w_{A^*(\nu)}
   \Big( \frac{1}{2} \ln(\frac{\mu^\nu_{A^*(\nu)}}{\mu^{\mathrm{u}}_{A^*(\nu)}}) + \frac{\mu^{\mathrm{u}}_{A^*(\nu)}}{2\mu^\nu_{A^*(\nu)}} \nonumber\\
   &+ \frac{(\mu^{\mathrm{u}}_{A^*(\nu)}-\mu^\nu_{A^*(\nu)})^2}{4\mu^\nu_{A^*(\nu)}\sigma^2} -\frac{1}{2} \big)
\nonumber\\&+   w_{i} \big(\frac{1}{2} \ln(\frac{\mu^\nu_{i}}{\mu^{\mathrm{u}}_{i}}) + \frac{\mu^{\mathrm{u}}_{i}}{2\mu^\nu_{i}} + \frac{(\mu^{\mathrm{u}}_{i}-\mu^\nu_{i})^2}{4\mu^\nu_{i}\sigma^2} -\frac{1}{2} \Big), \nonumber\\
 \text{s.t.}\;\;\; & 0<\mu^{\mathrm{u}}_{A^*(\nu)}  \le \mu^{\mathrm{u}}_{i},\label{eq1}
\end{align}
and obtain the optimal values of $\mu^{\mathrm{u}}_{A^*(\nu)},\mu^{\mathrm{u}}_{i}$, i.e.,
\begin{equation}\label{TP33}
\mu^{\mathrm{u}}_{A^*(\nu)} = \mu^{\mathrm{u}}_{i} = \frac{w_{A^*(\nu)} +  w_{i}}{w_{A^*(\nu)}\mu^\nu_{i} + w_{i} \mu^\nu_{A^*(\nu)}}\mu^\nu_{A^*(\nu)}\mu^\nu_{i}.
\end{equation}}
By substituting \eqref{TP33} into \eqref{TP32}, we have
\begin{align}
&c^*(\nu)^{-1}  = \sup_{\mathbf{w} \in \mathcal{W}_{I}}\inf_{i \neq A^*(\nu)}\Big\{
\frac{w_{A^*(\nu)}}{2}\ln\Big(\frac{w_{A^*(\nu)}\mu^\nu_{i} + w_{i} \mu^\nu_{A^*(\nu)}}{(w_{A^*(\nu)} +  w_{i})\mu^\nu_{i}}\Big)
\nonumber\\
&
+ \frac{w_{i}}{2}\ln\Big(\frac{w_{A^*(\nu)}\mu^\nu_{i} \!+\! w_{i}\mu^\nu_{A^*(\nu)}}{(w_{A^*(\nu)} \!+\! w_{i})\mu^\nu_{A^*(\nu)}}\Big)
  \!+\! \frac{w_{A^*(\nu)} w_{i}(\mu^\nu_{A^*(\nu)} \!-\! \mu^\nu_{i})^2) }{4\sigma^2 ( w_{A^*(\nu)}\mu^\nu_{i} \!+\! w_{i}\mu^\nu_{A^*(\nu)})} \nonumber\\
  &\qquad\qquad- \frac{1}{2}(w_{A^*(\nu)} + w_i)\Big\}.\label{T22}
\end{align}
With $x_i^* = \frac{\mu^\nu_{i}}{\mu^\nu_{A^*(\nu)}}$,   we have
\begin{equation}
w^*_{A^*(\nu)} = \frac{1}{1+\sum_{a\neq A^*(\nu)}x_a^* },\;\;\;
w^*_{i} = \frac{x_i^*}{1+\sum_{a\neq A^*(\nu)}x_a^*}.
\end{equation}
As a result, according to \eqref{T22}, $x_1^*,\ldots, x_I^*$ belongs to the set
\begin{equation}\label{T42}
\argmax_{x_1,\ldots, x_I} \min_{i \neq A^*(\nu)}\frac{1}{1+\sum_{a\neq A^*(\nu)}x_a} f_{Y,i}(x_i).
\end{equation}
where
\begin{align}
f_{Y,i}(x_i) &= \frac{1}{2}\ln\Big(\frac{\mu^\nu_{i} + x_i \mu^\nu_{A^*(\nu)}}{(1 +  x_{i})\mu^\nu_{i}}\Big)
+ \frac{x_i}{2}\ln\Big(\frac{\mu^\nu_{i} + x_i \mu^\nu_{A^*(\nu)}}{(1 +  x_{i})\mu^\nu_{A^*(\nu)}}\Big) \nonumber\\
&+  \frac{ x_i(\mu^\nu_{A^*(\nu)} - \mu^\nu_{i})^2 }{4\sigma^2 ( \mu^\nu_{i} + x_i \mu^\nu_{A^*(\nu)})} - \frac{1+x_i}{2}.
\end{align}
By differentiating $f_{Y,i}(x_i)$ with respect to $x_i$, we have
\begin{align}
\frac{\mathrm{d} f_{Y,i}}{\mathrm{d} x_i}
&=   \frac{\mu^\nu_{A^*(\nu)}-\mu^\nu_{i}}{2(\mu^\nu_{i}+x_i\mu^\nu_{A^*(\nu)})}
+ \frac{1}{2}\ln\Big(\frac{\mu^\nu_{i}+x_i\mu^\nu_{A^*(\nu)}}{(1+x_i)\mu^\nu_{A^*(\nu)} }\Big) \nonumber\\*
&+ \frac{\mu^\nu_i(\mu^\nu_{A^*(\nu)} - \mu^\nu_i)^2}{4\sigma^2( \mu^\nu_{i} + x_i \mu^\nu_{A^*(\nu)})^2}
-\frac{1}{2}.
\end{align}
Since $\frac{\mathrm{d}^2f_{Y,i}}{\mathrm{d} x_i^2}
\le 0$, $\frac{\mathrm{d}f_{Y,i}}{\mathrm{d} x_i}$ is monotonically decreasing for all $x_i\in [0,\infty)$. Furthermore, it holds that
$\lim\limits_{x_i\rightarrow \infty} \frac{\mathrm{d} f_{Y,i}}{\mathrm{d} x_i} = -1/2$. {Because of \eqref{sigma}, we have $\frac{\mathrm{d}f_{Y,i}}{\mathrm{d}x_i} \large|_{x_i = 0}\le 0$.} As such,
 $\frac{\mathrm{d}f_{Y,i}}{\mathrm{d}x_i}\le 0$ holds for all  $x_i \in [0,\infty)$, and $f_{Y,i}(x_i)$ is a monotonically decreasing function.
Then, we let
\begin{equation}
y^* = f_{Y,1}(x_1^*) = f_{Y,2}(x_2^*) = \ldots = f_{Y,I}(x_I^*),\;\;\;\; y^*\in [0, y_u],
\end{equation}
based on which we then introduce the inverse function of $f_{Y,i}(x)$, i.e. $f_{X,i} = f_{Y,i}^{-1}$.
According to \eqref{T42}, $y^*$ belongs to the set
\begin{equation}
\argmin_{y\in [0, y_u]} \mathcal{F}_Y(y)   \;\;\; \text{with} \;\;\;
\mathcal{F}_Y(y)  = \frac{y}
{1+\sum_{g\neq A^*(\nu)}f_{X,i}(y^*)}.
\end{equation}
 
By taking the derivative of $\mathcal{F}_Y(y)$, we have
\begin{equation}\label{T43}
\frac{\mathrm{d}\mathcal{F}_Y}{\mathrm{d}y} \!=\! \frac{1}{1\!+\!\sum_{i\neq A^*(\nu)} f_{X,i}(y) }
\!-\! \sum_{i\neq A^*(\nu)} \frac{yf'_{X,i}(y) }{(1\!+\!\sum_{i\neq A^*(\nu)} f_{X,i}(y)  )^2}.
\end{equation}
By setting $\frac{\mathrm{d}\mathcal{F}_Y}{\mathrm{d}y} = 0$, we obtain $y$ such that $\mathcal{F}_Y(y)$ achieves its minimum, i.e.,
\begin{equation}\label{T45}
\sum_{i\neq A^*(\nu)} {y f'_{X,i}(y) } -  f_{X,i}(y) = 1.
\end{equation}
Then we let $\mathcal{F}_{\nu}(y) = \sum_{i\neq A^*(\nu)} {y f'_{X,i}(y) } -  f_{X,i}(y)$. Since $f'_{X,i}(y) = \frac{1}{f'_{Y,i}(y)}$ and $\frac{\mathrm{d}^2f_{Y,i}}{\mathrm{d} y^2}\le 0$, we have
\begin{align}  \frac{\mathrm{d}\mathcal{F}_{\nu}(y)}{\mathrm{d}y}& = \sum_{i\neq A^*(\nu)}\frac{y \mathrm{d}^2 f_{X,i}(y)}{\mathrm{d}y^2} \nonumber\\
&= \sum_{i\neq A^*(\nu)} -\frac{y}{(f'_{Y,i}(y))^2}\frac{\mathrm{d}^2f_{Y,i}}{\mathrm{d}y^2}\ge 0.
\end{align}
 Therefore, since $\mathcal{F}_{\nu}(y)$ is   strictly increasing and $\mathcal{F}_{\nu}(y_u) = \infty$, $\mathcal{F}_{\nu}(0)  = 0$, $\mathcal{F}_{\nu}(y) = 1$ must have a unique solution $y^*$, which can be solved using  the bisection method. Then, we can obtain $w^*_{i}$ as
\begin{equation}
w^*_{i}=  \frac{f_{X,i}(y^*)}{\sum_{a=1}^I f_{X,a}(y^*) }.
\end{equation}
{Then, by replacing the  mean $\mu_{i}^\nu$ with the empirical mean $\hat{\mu}^\nu_{i}(t)$, we can obtain $\hat{w}^*_i(t)$ as in \eqref{upW}.}

\vspace{-0.4cm}
\section{}\label{PL9}
\subsection{Proof of Lemma \ref{SC1}}
\begin{proof}
Assume a heteroscedastic Gaussian bandit instance $\nu$ with $\mu^\nu_{1}\ge \mu^\nu_{2}\ge \ldots \ge \mu^\nu_{I}$.
There exists $\xi=\xi(\epsilon) \leq\left(\mu^\nu_{1}-\mu^\nu_{2}\right) / 4$ such that
\begin{equation}
\mathcal{I}_{\epsilon}:=\left[\mu^\nu_{1}-\xi, \mu^\nu_{1}+\xi\right] \times \ldots \times\left[\mu^\nu_{I}-\xi, \mu^\nu_{I}+\xi\right].
\end{equation}
Then, for a bandit model $\hat{\nu}_t \in \mathcal{I}_{\epsilon} $ and $t_0 >0$
\begin{equation}
\sup_{t\ge t_0}\max\limits_{i}|\hat{w}^*_{i}(t) - w^*_{i} |\le \epsilon.
\end{equation}
Furthermore, for all $\hat{\nu}_t \in \mathcal{I}_{\epsilon}$, the empirical optimal arm is $A^*(\hat{\nu}_t) =  1$.

Let define $h(T) =  T^{1/4}$ and the event
\begin{equation}
\mathcal{E}_{T}(\epsilon)= \bigcap_{t=h(T)}^{T} \{ \hat{\nu}_t \in \mathcal{I}_\epsilon \}
\end{equation}
in which it holds for $t\ge h(T)$, $A^*(\hat{\nu}_t)=1$.
Then, let us rewrite $Z(t)$ in \eqref{Z2} as
\begin{align}
Z(t) &=  \min_{i\neq 1} \Big({T_{{1}}(t)}
      D_{\mathrm{HG}}(\hat{\mu}^\nu_{1}(t),q_i(t))
      \!+\! {T_{i}(t)}D_{\mathrm{HG}}(\hat{\mu}^\nu_{i}(t),q_i(t))\Big)\nonumber\\
      &= t\ {f}_{\mathrm{Z}}\Big(\hat{\nu}_t,\Big(\frac{T_{i}(t)}{t}\Big)_{i=1}^I\Big),
\end{align}
where $q_i(t)$ is defined in~\eqref{qfun},  ${f}_{\mathrm{Z}}\big(\nu', \mathbf{w}'\big) =  \min_{i\neq 1} \big(w'_{1}
      D_{\mathrm{HG}}({\mu}'_{1},{q}'_{i})
      + w'_{i}D_{\mathrm{HG}}({\mu}'_{i},q'_{i}) \big)$
      and
\begin{equation}
q'_{i} = \frac{w'_{1}  +  w'_{i}}{w'_{1} {\mu}'_{i} + w'_{i} {\mu}'_{1}}{\mu}'_{1}\mu'_{i}.
\end{equation}
\begin{lemma}\label{lemmaS}
 The sampling rule ensures that $T_{i}(t)\ge \sqrt{t} -1$ and that for all $\epsilon > 0$ and $t_0 > 0$, there exists a constant $t_{\epsilon } = \max\big\{ \big\lceil\frac{t_0}{3\epsilon}\big\rceil, \big\lceil\frac{1}{3\epsilon^2}\big\rceil, \big\lceil \frac{1}{12\epsilon^3} \big\rceil  \big\}$ such that
 \begin{align}
 &\sup _{t \geq t_{0}} \max _{i}\left|\hat{w}^{*}_i(t)-w_{i}^{*}\right| \leq \epsilon \nonumber\\ 
 &\Longrightarrow \sup _{t \geq t_{\epsilon}} \max _{i}\left|\frac{T_{i}(t)}{t}-w_{i}^{*}\right| \leq 3(I-1) \epsilon.
 \end{align}
 \end{lemma}
\begin{proof}
See the proof in Appendix \ref{plemma2}.
\end{proof}
According to Lemma \ref{lemmaS} and the definition of $\mathcal{E}_T$, when $T\ge t_{\epsilon} =  \max\{ \lceil\frac{t_0}{3\epsilon}\rceil, \lceil\frac{1}{3\epsilon^2}\rceil, \lceil \frac{1}{12\epsilon^3} \rceil  \}$, we define
\begin{equation}\label{cv}
C^*_{\epsilon}(\nu) = \inf\limits_{\hat{\nu} \in \mathcal{I}_{\epsilon}, \hat{\mathbf{w}}: | \hat{{w}}_i -{{w}}_i^* |\le 3(I-1)\epsilon  }{f}_{\mathrm{Z}}\big(\hat{\nu}, \hat{\mathbf{w}}\big),
\end{equation}
then on the event $\mathcal{E}_T$ it is holds that
\begin{equation}
Z(t)\ge tC^*_{\epsilon}(\nu), \;\;\;  \forall \, t \ge \sqrt{T}.
\end{equation}
When $T\ge t_{\epsilon}$, and on the event $\mathcal{E}_T$, it holds that
\begin{align}
&\min \big\{ \tau_{{\delta}},T\big\}\nonumber\\
& \le \sqrt{T} +
\sum_{t = \sqrt{T}}^T\mathbb{I}_{(\tau_{{\delta}}>t)} \le
\sqrt{T} + \sum_{t = \sqrt{T}}^T\mathbb{I}_{(Z(t)\le \beta(t,{\delta},\alpha))}
\nonumber\\&\le   \sqrt{T} + \sum_{t = \sqrt{T}}^T\mathbb{I}_{(tC^*_{\epsilon}(\nu)\le \beta(T,{\delta},\alpha))} \nonumber\\
&= \max\Big\{\sqrt{T},\frac{\beta(T,{\delta},\alpha)}{C^*_{\epsilon}(\nu)}\Big\}
\end{align}
Define
\begin{align}
 T_0({\delta}) & = \inf\Big\{T \in \mathbb{N}:  \max\big\{\sqrt{T},\frac{\beta(T,{\delta},\alpha)}{C^*_{\epsilon}(\nu)}\big\}\le T  \Big\} \nonumber\\
 &=   \inf\Big\{T \in \mathbb{N}: \frac{\beta(T,{\delta},\alpha)}{C^*_{\epsilon}(\nu)} \le T \Big\}  \nonumber\\
&= \inf \Big\{T \in \mathbb{N}: {C_{\epsilon}^{*}(\nu)} T \geq \ln ({\alpha T}/{{\delta}})\Big\}.
\end{align}
Using   \cite[Lemma 18]{Garivier2016}, we have
\begin{equation}
T_{0}({\delta}) \leq \frac{\alpha}{C^*_{\epsilon}(\nu)}\Big[\ln \big(\frac{ \alpha e}{{\delta} C_{\epsilon}^{*}(\nu)}\big)+\ln \ln \big(\frac{\alpha}{{\delta} C_{\epsilon}^{*}(\nu)}\big)\Big].
\end{equation}

\begin{lemma}\label{lemma5}
There exist two constants $\Gamma_2$ and $\Gamma_3$ (which depend on $\nu$ and $\epsilon$) such that
\begin{equation}
\mathbb{P}(\mathcal{E}_T^c) \le \Gamma_2 T\exp(-\Gamma_3 T^{1/8}).
\end{equation}
\end{lemma}
\begin{proof}
See the proof in Appendix \ref{plemma1}.
\end{proof}
{Using Lemma \ref{lemma5}, for every $T \ge \max(T_0({\delta}),t_{\epsilon})$, one has $\mathcal{E}_T \subset \{\tau_{{\delta}}\le T\}$, therefore
\begin{equation}
\mathbb{P}\left(\tau_{{\delta}}>T\right) \leq \mathbb{P}\left(\mathcal{E}_{T}^{c}\right) \leq \Gamma_2 T \exp \left(-\Gamma_3 T^{1 / 8}\right).
\end{equation}
According to the definition of $C_{\epsilon}^{*}(\nu)$ in \eqref{cv} and $c^*(\nu)^{-1}$  in~\eqref{Theo2}, it holds that $C_{\epsilon}^{*}(\nu) \le c^*(\nu)^{-1}$.
As a result, we have
\begin{align}
\mathbb{E}[\tau_{{\delta}}] &= \sum_{T=1}^{\infty} \mathbb{P}(\tau_{{\delta}}\ge T) \nonumber\\
 &= \sum_{T=1}^{\max(T_0({\delta}),t_{{\epsilon}})} \mathbb{P}(\tau_{{\delta}}\ge T) \!+  \! \sum_{T=\max(T_0({\delta}),t_{{\epsilon}})+1}^{\infty}\! \mathbb{P}(\tau_{{\delta}}\ge T) \nonumber\\
  &\leq t_{\epsilon}+T_{0}({\delta})+\sum_{T=1}^{\infty} \Gamma_2 T \exp \left(-\Gamma_3 T^{1 / 8}\right)\nonumber\\
  &\le  t_{\epsilon} + \alpha c^*(\nu)\Big[\ln \big(\frac{ \alpha e c^*(\nu)}{{\delta}}\big)+\ln \ln \big(\frac{\alpha c^*(\nu)}{{\delta} }\big)\Big] \nonumber\\
  &\qquad+\sum_{T=1}^{\infty} \Gamma_2 T \exp \left(-\Gamma_3 T^{1 / 8}\right).
 \end{align}

\begin{lemma}\label{Theom3}
 Let $\nu$ represent a heteroscedastic Gaussian bandit instance.
When $\nu$ has a unique optimal arm, then we have
\begin{equation}
c^*(\nu) \le c^*_{\mathrm{u}}(\nu),
\end{equation}
where
\begin{align}
c^*_{\mathrm{u}}(\nu) &= \Big( \frac{\mu^\nu_{A^*(\nu)}}{2\sum_{i=1}^I\mu^\nu_{i}}\ln\big(\frac{2 \mu^\nu_{A^*(\nu)}}{ \mu^\nu_{A^*(\nu)} + \mu^\nu_{A'(\nu)}} \big)
\nonumber\\ 
&+ \frac{ \mu^\nu_{A'(\nu)}}{2\sum_{i=1}^I\mu^\nu_{i}}\ln\big(\frac{2  \mu^\nu_{A'(\nu)}}{ \mu^\nu_{A^*(\nu)} +  \mu^\nu_{A'(\nu)}} \big) \nonumber\\
&+\frac{(\mu^\nu_{A^*(\nu)} - \mu^\nu_{A'(\nu)})^2}{8\sigma^2\sum_{a=1}^I\mu^\nu_{a} } - \frac{\mu^\nu_{A^*(\nu)}+\mu^\nu_{A'(\nu)}}{2\sum_{a=1}^I\mu^\nu_{a} }
\Big)^{-1}.
\end{align}
\end{lemma}
\begin{proof}
Please refer to Appendix \ref{plemma3}.
\end{proof}
Let $\mathcal{V}$ represent a set of $I$-armed heteroscedastic Gaussian bandit instances.
Note that
\begin{align}
c^*(\nu)^{-1}\! =\!  \sup_{\mathbf{w}\in \mathcal{W}_{I}}\inf_{\substack{ \mathrm{u} \in \mathcal{V}:\\ i \neq A^*(\nu), \mu^{\mathrm{u}}_{i}> \mu^{\mathrm{u}}_{A^*(\nu)}}}\!\! \Big(\!\sum_{a\in \{A^*(\nu),i\}}
w_{i}
D_{\mathrm{HG}}(\mu^\nu_{a},\mu^{\mathrm{u}}_{a})\Big),
\end{align}
which is related to the choice of $\mathrm{u}$ and $\mathbf{w}$.
According to  Lemma~\ref{Theom3}, it holds that
\begin{align}
\mathbb{E}[\tau_{{\delta}}] &\le  t_{\epsilon} + \alpha c^*(\nu)\Big[\ln \big(\frac{ \alpha e c^*(\nu)}{{\delta}}\big)+\ln \ln \big(\frac{\alpha c^*(\nu)}{{\delta} }\big)\Big] \nonumber\\ &\qquad+\sum_{T=1}^{\infty} \Gamma_2 T \exp (-\Gamma_3 T^{1 / 8})\\
&\le t_{\epsilon} + \alpha c^*_{\mathrm{u}}(\nu)\Big[\ln \big(\frac{ \alpha e c^*_{\mathrm{u}}(\nu)}{{\delta}}\big)+\ln \ln \big(\frac{\alpha c^*_{\mathrm{u}}(\nu)}{{\delta} }\big)\Big]\nonumber\\* 
&\qquad+\sum_{T=1}^{\infty} \Gamma_2 T \exp (-\Gamma_3 T^{1 / 8}).
\end{align}
Since $\tau = I+\tau_{{\delta}}$ according to Line \ref{tau1} of Algorithm \ref{alg2}, it holds that
\begin{align}
\mathbb{E}[\tau]& = I+ \mathbb{E}[\tau_{\delta}] \nonumber\\
& \le
 I+ t_{\epsilon} + \alpha c^*_{\mathrm{u}}(\nu)\Big[\ln \big(\frac{ \alpha e c^*_{\mathrm{u}}(\nu)}{{\delta}}\big)+\ln \ln \big(\frac{\alpha c^*_{\mathrm{u}}(\nu)}{{\delta} }\big)\Big]\nonumber\\ 
 &\qquad+\sum_{T=1}^{\infty} \Gamma_2 T \exp (-\Gamma_3 T^{1 / 8}).
  \end{align}
The proof is concluded by letting $A=I+T_\epsilon$.
}
\end{proof}

\vspace{-0.4cm}
\subsection{Proof of Lemma \ref{lemmaS}}\label{plemma2}

\begin{proof}
 Fix $t_0'$ such that when $t'_0\ge t_0$, it holds that
 \begin{equation}
 \forall\, t \geq t'_{0}, \quad \sqrt{t} \leq 2 t \epsilon \text { and } 1 / t \leq \epsilon.
 \end{equation}
 To satisfy the requirement $\sqrt{t} \leq 2 t \epsilon$, we fix  $t\ge \max\{t_0, \frac{1}{4\epsilon^2}\}$.
Thus, we let $t'_0:= \max\{t_0, \frac{1}{\epsilon}, \frac{1}{4\epsilon^2}  \}$.
 According to \cite[Lemma~17]{Garivier2016} and its proof in  \cite[Appendix~B.2]{Garivier2016},
 by choosing $\hat{\lambda}(i) = \hat{w}^*_{i}(t)$ and
 $\bm{\lambda}^* = \mathbf{w}^*$, we have
 \begin{align}
 \sup_{i}\Big|\frac{T_{i}(t)}{t}-w_{i}^{*}(\nu)\Big| &\leq(I-1) \max \Big\{2 \epsilon+\frac{1}{t}, \frac{t'_{0}}{t}\Big\}\nonumber\\ &\leq(I-1) \max \Big\{3 \epsilon, \frac{t'_{0}}{t}\Big\}.
 \end{align}
As a result, when $t \ge \frac{t_0}{3\epsilon}$, it holds that $\sup_{i}\big|\frac{T_{i}(t)}{t}-w_{i}^{*}(\nu)\big| \le 3(I-1)\epsilon$. Let 
$
t_{\epsilon } = \max\big\{ \big\lceil\frac{t_0}{3\epsilon}\big\rceil, \big\lceil\frac{1}{3\epsilon^2}\big\rceil, \big\lceil \frac{1}{12\epsilon^3} \big\rceil  \big\},
$
which concludes the proof.
 \end{proof}

\subsection{Proof of Lemma \ref{lemma5}}\label{plemma1}
\begin{proof}
First, we have
\begin{align}
\mathbb{P}(\mathcal{E}_T^c) &\le \sum_{t = h(T)}^T\mathbb{P}(\hat{\nu}_t\notin \mathcal{I}_{\epsilon} ) \nonumber\\
&=
\sum_{t = h(T)}^T\sum_{i = 1}^I[\mathbb{P}(\hat{\mu}^\nu_{i}(t)\le \mu^\nu_{i}-\xi) +
\mathbb{P}(\hat{\mu}^\nu_{i}(t)\ge \mu^\nu_{i}+\xi)].
\end{align}
According to Lemma \ref{lemmaS}, for each arm, we have
{
\begin{equation}\label{lp51}
T_{i}(t)>  \sqrt{t} - I, \qquad \forall\, t\ge h(T).
\end{equation}
}
Then,  we have
\begin{align}
\mathbb{P}\left(\hat{\mu}^\nu_{i}(t) \leq \mu^\nu_{i}-\xi\right) &\overset{(a)}{=}\mathbb{P}(\hat{\mu}^\nu_{i}(t) \leq \mu^\nu_{i}-\xi, T_{i}(t) \geq \sqrt{t}-I)\nonumber\\ &\overset{(b)}{\leq} \sum_{m=\sqrt{t}-I}^{t} \mathbb{P}\left(\hat{\mu}_{i}^\nu(m) \leq \mu^\nu_{i}-\xi\right) \nonumber\\
& \overset{(c)}{\leq} \sum_{m=\sqrt{t}-I}^{t} \exp \left(-m D_{\mathrm{HG}}\left(\mu^\nu_{i}-\xi, \mu^\nu_{i}\right)\right) \nonumber\\ 
&\leq \frac{e^{-(\sqrt{t}-I) D_{\mathrm{HG}}\left(\mu^\nu_{i}-\xi, \mu^\nu_{i}\right)}}{1-e^{-D_{\mathrm{HG}}\left(\mu^\nu_{i}-\xi, \mu^\nu_{i}\right)}} .
\end{align}
where $(a)$ is obtained due to  \eqref{lp51}, $(b)$ and $(c)$  holds due to the union bound and the Chernoff inequality. Due to the same reason, it holds that
\begin{equation}
\mathbb{P}\left(\hat{\mu}^\nu_{i}(t) \ge  \mu^\nu_{i}+\xi\right)  \le \frac{e^{-(\sqrt{t}-I) D_{\mathrm{HG}}\left(\mu^\nu_{i}+\xi, \mu^\nu_{i}\right)}}{1-e^{-D_{\mathrm{HG}}\left(\mu^\nu_{i}+\xi, \mu^\nu_{i}\right)}}.
\end{equation}
Then, define
\begin{equation}
\begin{aligned}
 \Gamma_2 &:= \sum_{i=1}^I\Big(\frac{e^{ID_{\mathrm{HG}}(\mu^\nu_{i}-\xi,\mu^\nu_{i})}}{1-e^{-D_{\mathrm{HG}}(\mu^\nu_{i}-\xi,\mu^\nu_{i})}} +  \frac{e^{ID_{\mathrm{HG}}(\mu^\nu_{i}+\xi,\mu^\nu_{i})}}{1-e^{-D_{\mathrm{HG}}(\mu^\nu_{i}+\xi,\mu^\nu_{i})}}  \Big)\\
\Gamma_3 &:= \min\limits_{i\in [I]}\big[\min\big\{ D_{\mathrm{HG}}(\mu^\nu_{i}-\xi, \mu^\nu_{i}), D_{\mathrm{HG}}(\mu^\nu_{i}+\xi, \mu^\nu_{i})\big\}\big],
\end{aligned}
\end{equation}
 we obtain
\begin{align}
\mathbb{P}(\mathcal{E}_T^c) &\le \sum_{t = h(T)}^T\sum_{i = 1}^I\big[\mathbb{P}(\hat{\mu}^\nu_{i}(t)\le \mu^\nu_{i}-\xi) +
\mathbb{P}(\hat{\mu}^\nu_{i}(t)\le \mu^\nu_{i}+\xi)\big]\nonumber\\
&\le \sum_{t=h(T)}^{T} \Gamma_2\exp (-\sqrt{t} \Gamma_3) \leq \Gamma_2 T \exp (-\sqrt{h(T)} \Gamma_3) \nonumber\\ 
&=\Gamma_2 T \exp \big(-\Gamma_3 T^{1 / 8}\big),
\end{align}
which concludes the proof.
\end{proof}

\subsection{Proof of Lemma \ref{Theom3}}\label{plemma3}
\begin{proof}
By setting $\hat{\mathbf{w}}\in \mathcal{W}_I$ as $\hat{w}_i = \frac{\mu^\nu_{i}}{\sum_{i=1}^I\mu^\nu_{i}}$, according to~\eqref{T22}, $c^*(\nu)^{-1}$ satisfies
\begin{align}
\!\!\!\! c^*(\nu)^{-1} \ge  \inf_{i \neq A^*(\nu)}& \Big\{
\frac{\mu^\nu_{A^*(\nu)}}{2\sum_{a=1}^I\mu^\nu_{a}}\ln\big(\frac{2 \mu^\nu_{A^*(\nu)}}{ \mu^\nu_{A^*(\nu)} + \mu^\nu_{i}} \big)
\nonumber\\
&+ \frac{\mu^\nu_{i}}{2\sum_{a=1}^I\mu^\nu_{a}}\ln\big(\frac{2 \mu^\nu_{i}}{ \mu^\nu_{A^*(\nu)} + \mu^\nu_{i}} \big) \nonumber\\
&+
\frac{(\mu^\nu_{A^*(\nu)} - \mu^\nu_{i})^2}{8\sigma^2\sum_{a=1}^I\mu^\nu_{a} } - \frac{\mu^\nu_{A^*(\nu)}+\mu^\nu_i}{2\sum_{a=1}^I\mu^\nu_{a} } \Big\}.\label{eqc-1}
\end{align}
Define 
\begin{align}
\mathcal{F}_\mu(\mu^\nu_{i})& \triangleq  \frac{\mu^\nu_{A^*(\nu)}}{2\sum_{a=1}^I\mu^\nu_{a}}\ln\big(\frac{2 \mu^\nu_{A^*(\nu)}}{ \mu^\nu_{A^*(\nu)} + \mu^\nu_{i}} \big)
 \nonumber\\ 
 &\quad + \frac{\mu^\nu_{i}}{2\sum_{a=1}^I\mu^\nu_{a}}\ln\big(\frac{2 \mu^\nu_{i}}{ \mu^\nu_{A^*(\nu)} + \mu^\nu_{i}} \big) \nonumber\\
&\quad + \frac{(\mu^\nu_{A^*(\nu)} - \mu^\nu_{i})^2}{8\sigma^2\sum_{a=1}^I\mu_{a} } - \frac{(\mu^\nu_{A^*(\nu)}+\mu^\nu_i)}{2\sum_{a=1}^I\mu^\nu_{a} } .
\end{align}
By taking the derivative of $\mu_{A^*(\nu)}$, we have
\begin{align}
 \frac{\mathrm{d}\mathcal{F}_\mu(\mu^\nu_{i})}{\mathrm{d} \mu^\nu_{i}}& =  -\frac{\mu^\nu_{A^*(\nu)}}{2(\sum_{a=1}^I\mu^\nu_{a})^2 }\ln\Big(\frac{2 \mu^\nu_{A^*(\nu)}}{\mu^\nu_{A^*(\nu)} + \mu^\nu_{i}}  \Big)  \nonumber\\
&\quad + \frac{\sum_{i\neq g} \mu^\nu_{i }}{2(\sum_{a=1}^I \mu^\nu_{a})^2}\ln\Big(\frac{2\mu^\nu_{i}}{\mu^\nu_{i}+ \mu^\nu_{A^*(\nu)}}  \Big)  \nonumber\\
 &\quad + \frac{(\mu^\nu_{i}-2)\sum_{a=1}^I \mu^\nu_{a}-(\mu^\nu_{A^*(\nu)} - \mu^\nu_{i})^2}
 {8\sigma^2(\sum_{a=1}^I \mu^\nu_{a})^2} \nonumber\\
 &\quad -
 \frac{\sum_{a=1}^I \mu^\nu_{a} - \mu^\nu_{A^*(\nu)} - \mu^\nu_i}{2(\sum_{a=1}^I \mu^\nu_{a})^2}.
 \end{align}
Since $\frac{\mathrm{d}\mathcal{F}_\mu(\mu^\nu_{i})}{\mathrm{d} \mu^\nu_{i}}< 0$   holds, $\mathcal{F}_\mu(\mu^\nu_{i})$ is monotonically decreasing.
Therefore, we   conclude that when a suboptimal arm $A'(\nu) = \argmin_{i\neq A^*(\nu)} \mu^\nu_{i}$ is selected,
{\eqref{eqc-1} can be rewritten as
\begin{align} 
c^*(\nu)^{-1} &\ge \frac{\mu^\nu_{A^*(\nu)}}{2\sum_{a=1}^I\mu^\nu_{a}}\ln\big(\frac{2 \mu^\nu_{A^*(\nu)}}{ \mu^\nu_{A^*(\nu)} + \mu^\nu_{A'(\nu)}} \big) \nonumber\\
&\quad + \frac{ \mu^\nu_{A'(\nu)}}{2\sum_{a=1}^I\mu^\nu_{a}}\ln\big(\frac{2  \mu^\nu_{A'(\nu)}}{ \mu^\nu_{A^*(\nu)} +  \mu^\nu_{A'(\nu)}} \big)\nonumber\\
&\quad 
+
\frac{(\mu^\nu_{A^*(\nu)} - \mu^\nu_{A'(\nu)})^2}{8\sigma^2\sum_{a=1}^I\mu^\nu_{a} } - \frac{\mu^\nu_{A^*(\nu)}+\mu^\nu_{A'(\nu)}}{2\sum_{a=1}^I\mu^\nu_{a} },
\end{align}  as such
we have
\begin{align}
c^*(\nu) & \le \Big( \frac{\mu^\nu_{A^*(\nu)}}{2\sum_{a=1}^I\mu^\nu_{a}}\ln\big(\frac{2 \mu^\nu_{A^*(\nu)}}{ \mu^\nu_{A^*(\nu)} + \mu^\nu_{A'(\nu)}} \big)\nonumber\\
&
\quad + \frac{ \mu^\nu_{A'(\nu)}}{2\sum_{a=1}^I\mu^\nu_{a}}\ln\big(\frac{2  \mu^\nu_{A'(\nu)}}{ \mu^\nu_{A^*(\nu)} +  \mu^\nu_{A'(\nu)}} \big) \nonumber\\
&\quad +
\frac{(\mu^\nu_{A^*(\nu)} - \mu^\nu_{A'(\nu)})^2}{8\sigma^2\sum_{a=1}^I\mu^\nu_{a} } - \frac{\mu^\nu_{A^*(\nu)}+\mu^\nu_{A'(\nu)}}{2\sum_{a=1}^I\mu^\nu_{a} }
\Big)^{-1},
\end{align}
}
which concludes the proof.
\end{proof}

}

\end{appendices}

\vspace{-0.4cm}
\bibliography{BeamAlig}

\begin{thebibliography}{10}
\providecommand{\url}[1]{#1}
\csname url@samestyle\endcsname
\providecommand{\newblock}{\relax}
\providecommand{\bibinfo}[2]{#2}
\providecommand{\BIBentrySTDinterwordspacing}{\spaceskip=0pt\relax}
\providecommand{\BIBentryALTinterwordstretchfactor}{4}
\providecommand{\BIBentryALTinterwordspacing}{\spaceskip=\fontdimen2\font plus
\BIBentryALTinterwordstretchfactor\fontdimen3\font minus
  \fontdimen4\font\relax}
\providecommand{\BIBforeignlanguage}[2]{{%
\expandafter\ifx\csname l@#1\endcsname\relax
\typeout{** WARNING: IEEEtran.bst: No hyphenation pattern has been}%
\typeout{** loaded for the language `#1'. Using the pattern for}%
\typeout{** the default language instead.}%
\else
\language=\csname l@#1\endcsname
\fi
#2}}
\providecommand{\BIBdecl}{\relax}
\BIBdecl

\bibitem{wei2022}
Y.~Wei, Z.~Zhong, V.~Y.~F. Tan, and C.~Wang, ``Fast beam alignment via pure
  exploration in multi-armed bandits,'' in \emph{IEEE International Symposium
  on Information Theory (ISIT)}, 2022, pp. 1886--1891.

\bibitem{8048526}
P.~Zhou, X.~Fang, Y.~Fang, Y.~Long, R.~He, and X.~Han, ``Enhanced random access
  and beam training for millimeter wave wireless local networks with high user
  density,'' \emph{IEEE Trans. Wireless Commun.}, vol.~16, no.~12, pp.
  7760--7773, 2017.

\bibitem{8458146}
M.~Giordani, M.~Polese, A.~Roy, D.~Castor, and M.~Zorzi, ``A tutorial on beam
  management for {3GPP NR} at mm{W}ave frequencies,'' \emph{IEEE Commun.
  Surveys Tuts.}, vol.~21, no.~1, pp. 173--196, 2019.

\bibitem{7947209}
S.~Noh, M.~D. Zoltowski, and D.~J. Love, ``Multi-resolution codebook and
  adaptive beamforming sequence design for millimeter wave beam alignment,''
  \emph{IEEE Trans. Wireless Commun.}, vol.~16, no.~9, pp. 5689--5701, 2017.

\bibitem{7990158}
J.~Zhang, Y.~Huang, Q.~Shi, J.~Wang, and L.~Yang, ``Codebook design for beam
  alignment in millimeter wave communication systems,'' \emph{IEEE Trans.
  Commun.}, vol.~65, no.~11, pp. 4980--4995, 2017.

\bibitem{7959169}
M.~Xiao, S.~Mumtaz, Y.~Huang, L.~Dai, Y.~Li, M.~Matthaiou, G.~K. Karagiannidis,
  E.~Bjornson, K.~Yang, C.-L. I, and A.~Ghosh, ``Millimeter wave communications
  for future mobile networks,'' \emph{IEEE J. Sel. Areas Commun.}, vol.~35,
  no.~9, pp. 1909--1935, 2017.

\bibitem{8368998}
M.~Gao, B.~Ai, Y.~Niu, Z.~Zhong, Y.~Liu, G.~Ma, Z.~Zhang, and D.~Li, ``Dynamic
  mmwave beam tracking for high speed railway communications,'' in \emph{IEEE
  WCNCW}, 2018, pp. 278--283.

\bibitem{7390019}
Z.~Marzi, D.~Ramasamy, and U.~Madhow, ``Compressive channel estimation and
  tracking for large arrays in mm{W}ave picocells,'' \emph{IEEE J. Select.
  Topics Signal Process.}, vol.~10, no.~3, pp. 514--527, 2016.

\bibitem{5262295}
J.~Wang, Z.~Lan, C.~Pyo, T.~Baykas, C.~Sum, M.~Rahman, J.~Gao, R.~Funada,
  F.~Kojima, H.~Harada, and S.~Kato, ``Beam codebook based beamforming protocol
  for {multi-Gbps} millimeter-wave {WPAN} systems,'' \emph{IEEE J. Select.
  Areas Commun.}, vol.~27, no.~8, pp. 1390--1399, 2009.

\bibitem{8852637}
G.~E. Garcia, N.~Garcia, G.~Seco-Granados, E.~Karipidis, and H.~Wymeersch,
  ``Fast in-band position-aided beam selection in millimeter-wave {MIMO},''
  \emph{IEEE Access}, vol.~7, pp. 142\,325--142\,338, 2019.

\bibitem{8114345}
A.~Ali, N.~Gonzalez-Prelcic, and R.~W. Heath, ``Millimeter wave beam-selection
  using out-of-band spatial information,'' \emph{IEEE Trans. Wireless Commun.},
  vol.~17, no.~2, pp. 1038--1052, 2018.

\bibitem{7786130}
J.~Choi, V.~Va, N.~Gonzalez-Prelcic, R.~Daniels, C.~R. Bhat, and R.~W. Heath,
  ``Millimeter-wave vehicular communication to support massive automotive
  sensing,'' \emph{IEEE Commun. Mag.}, vol.~54, no.~12, pp. 160--167, 2016.

\bibitem{4723352}
Q.~Zhao, B.~Krishnamachari, and K.~Liu, ``On myopic sensing for multi-channel
  opportunistic access: structure, optimality, and performance,'' \emph{IEEE
  Trans.Wireless Commun.}, vol.~7, no.~12, pp. 5431--5440, 2008.

\bibitem{5535151}
K.~Liu and Q.~Zhao, ``Distributed learning in multi-armed bandit with multiple
  players,'' \emph{IEEE Trans. Signal Process.}, vol.~58, no.~11, pp.
  5667--5681, 2010.

\bibitem{5518773}
H.~I. Volos and R.~M. Buehrer, ``Cognitive engine design for link adaptation:
  An application to multi-antenna systems,'' \emph{IEEE Trans. Wireless
  Commun.}, vol.~9, no.~9, pp. 2902--2913, 2010.

\bibitem{6574205}
N.~Gulati and K.~R. Dandekar, ``Learning state selection for reconfigurable
  antennas: A multi-armed bandit approach,'' \emph{IEEE Trans. Antennas
  Propag.}, vol.~62, no.~3, pp. 1027--1038, 2014.

\bibitem{8723104}
M.~B. Booth, V.~Suresh, N.~Michelusi, and D.~J. Love, ``Multi-armed bandit beam
  alignment and tracking for mobile millimeter wave communications,''
  \emph{IEEE Commun. Lett.}, vol.~23, no.~7, pp. 1244--1248, 2019.

\bibitem{8472783}
G.~H. Sim, S.~Klos, A.~Asadi, A.~Klein, and M.~Hollick, ``An online
  context-aware machine learning algorithm for {5G} stochastic multi-armed
  bandits mm{W}ave vehicular communications,'' \emph{IEEE/ACM Trans. Netw.},
  vol.~26, no.~6, pp. 2487--2500, 2018.

\bibitem{8486279}
M.~Hashemi, A.~Sabharwal, C.~Emre~Koksal, and N.~B. Shroff, ``Efficient beam
  alignment in millimeter wave systems using contextual bandits,'' in
  \emph{IEEE INFOCOM 2018 - IEEE Conference on Computer Communications}, 2018,
  pp. 2393--2401.

\bibitem{8842625}
W.~Wu, N.~Cheng, N.~Zhang, P.~Yang, W.~Zhuang, and X.~Shen, ``Fast mmwave beam
  alignment via correlated bandit learning,'' \emph{IEEE Trans. Wireless
  Commun.}, vol.~18, no.~12, pp. 5894--5908, 2019.

\bibitem{8662770}
V.~Va, T.~Shimizu, G.~Bansal, and R.~W. Heath, ``Online learning for
  position-aided millimeter wave beam training,'' \emph{IEEE Access}, vol.~7,
  pp. 30\,507--30\,526, 2019.

\bibitem{9013578333}
M.~Hussain and N.~Michelusi, ``Second-best beam-alignment via {Bayesian}
  multi-armed bandits,'' in \emph{IEEE GLOBECOM}, 2019, pp. 1--6.

\bibitem{7967837}
S.~He, J.~Wang, Y.~Huang, B.~Ottersten, and W.~Hong, ``Codebook-based hybrid
  precoding for millimeter wave multiuser systems,'' \emph{IEEE Trans. Signal
  Proces.}, vol.~65, no.~20, pp. 5289--5304, 2017.

\bibitem{9689054}
H.~Zhang, Y.~Zhang, J.~Cosmas, N.~Jawad, W.~Li, R.~Muller, and T.~Jiang,
  ``mm{Wave} indoor channel measurement campaign for {5G} new radio indoor
  broadcasting,'' \emph{IEEE Trans. Broadcast.}, vol.~68, no.~2, pp. 331--344,
  2022.

\bibitem{7857002}
X.~Wu, C.-X. Wang, J.~Sun, J.~Huang, R.~Feng, Y.~Yang, and X.~Ge, ``60-{GHz}
  millimeter-wave channel measurements and modeling for indoor office
  environments,'' \emph{IEEE Trans. Antennas Propag.}, vol.~65, no.~4, pp.
  1912--1924, 2017.

\bibitem{7109864}
T.~S. Rappaport, G.~R. MacCartney, M.~K. Samimi, and S.~Sun, ``Wideband
  millimeter-wave propagation measurements and channel models for future
  wireless communication system design,'' \emph{IEEE Trans. Commun.}, vol.~63,
  no.~9, pp. 3029--3056, Sept. 2015.

\bibitem{Garivier2016}
A.~Garivier and E.~Kaufmann, ``Optimal best arm identification with fixed
  confidence,'' \emph{Journal of Machine Learning Research}, vol.~49, pp.
  998--1027, 2016.

\bibitem{6758357}
S.~Jaeckel, L.~Raschkowski, K.~Borner, and L.~Thiele, ``Quadriga: {A} 3-{D}
  multi-cell channel model with time evolution for enabling virtual field
  trials,'' \emph{IEEE Trans. Antennas Propag.}, vol.~62, no.~6, pp.
  3242--3256, Jun. 2014.

\bibitem{Lipschitz}
S.~Magureanu, R.~Combes, and A.~Proutiere, ``Lipschitz bandits: Regret lower
  bounds and optimal algorithms,'' in \emph{PMLR}, vol.~35, no. 975-999, 2014.

\bibitem{gibbs1902elementary}
J.~W. Gibbs, \emph{Elementary principles in statistical mechanics developed
  with special reference to the rational foundations of thermodynamics}.\hskip
  1em plus 0.5em minus 0.4em\relax Franklin Classics, 2018.

\end{thebibliography}
\bibliographystyle{IEEEtran}

\begin{IEEEbiography}[{\includegraphics[width=1in,height=1.25in,clip,keepaspectratio]{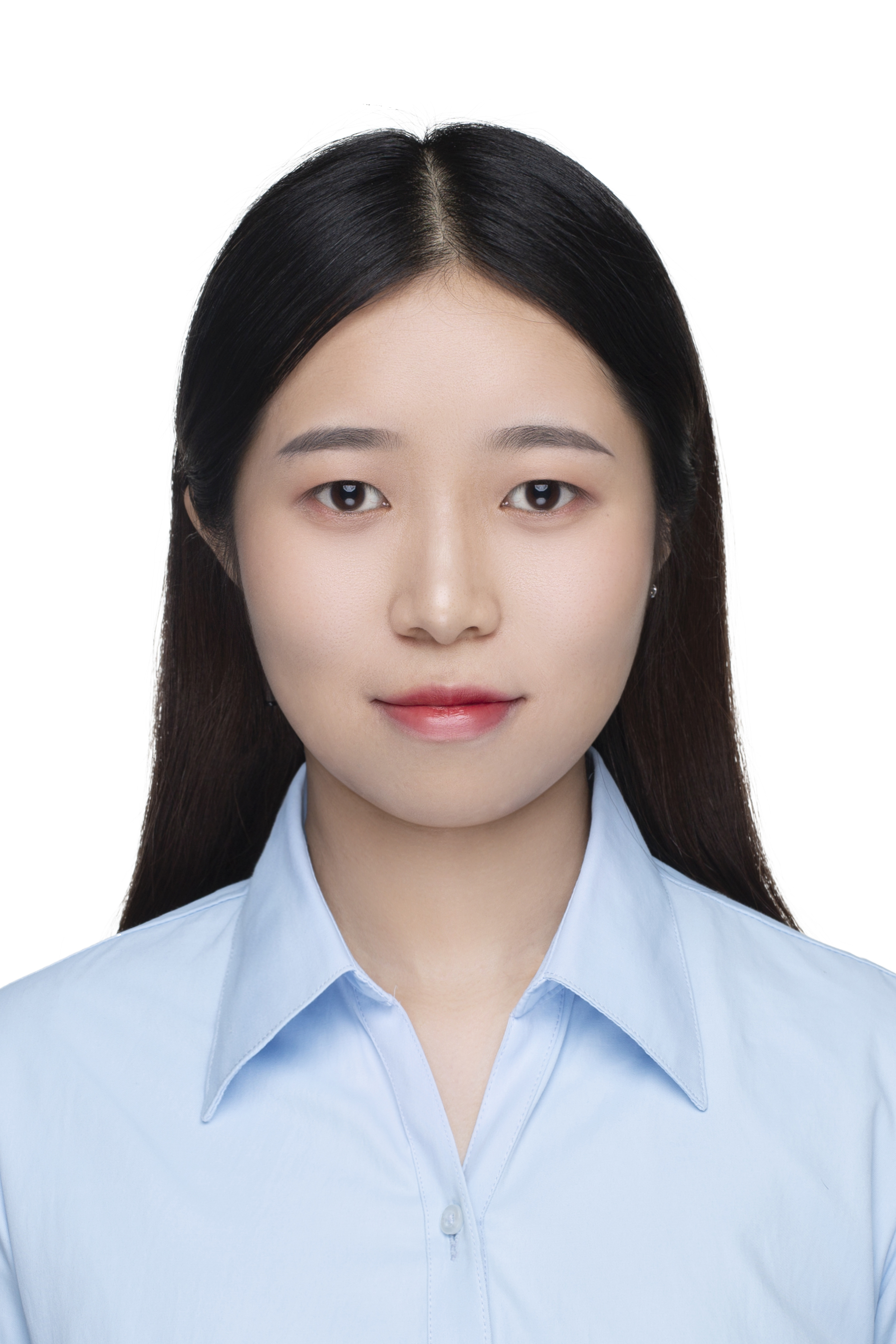}}]
 {Yi Wei}  received the B.Eng.\ and the Ph.D.\ degrees in Information and Communication Engineering from Zhejiang University in 2017 and 2022, respectively.
 From May 2021 to Apr 2022, she was a Visiting Scholar at the
Department of Electrical and Computer Engineering, National University of Singapore, Singapore.
 She is currently working as a Research Engineer at China Aerospace Science and Technology Corporation.
Her research interests include signal processing for wireless communications,  algorithm design for advanced MIMO and  deep learning for wireless communications.
\end{IEEEbiography}
 
\begin{IEEEbiography}[{\includegraphics[width=1in,height=1.25in,clip,keepaspectratio]{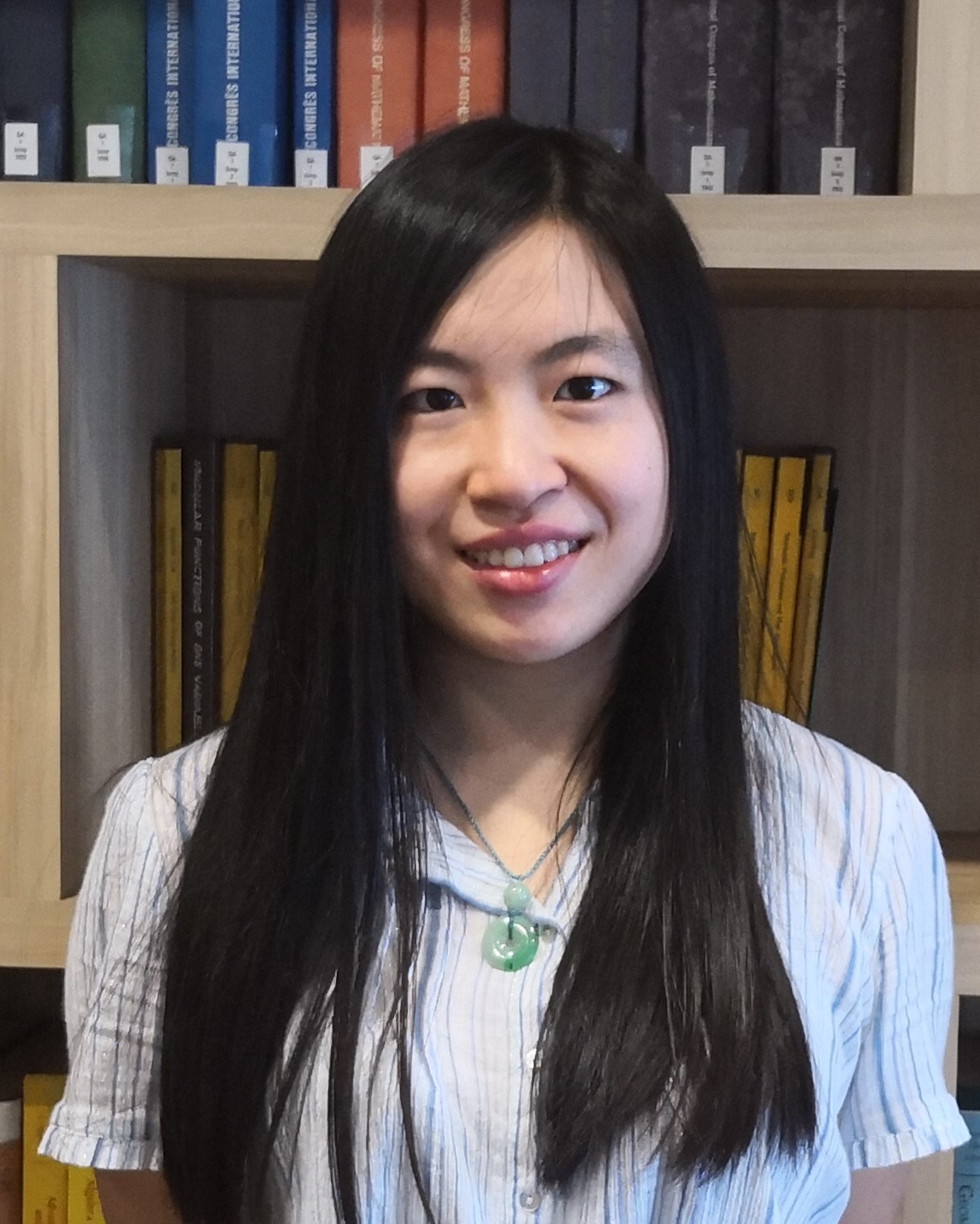}}]{Zixin Zhong} was born in China in 1995. She is currently a postdoctoral fellow at the Department of Computing Science of University of Alberta (UofA). She is supervised by Prof.\ Csaba Szepesv{\'a}ri. Dr.\ Zhong received her PhD degree from the Department of Mathematics of National University of Singapore (NUS) in October 2021. Dr.\ Zhong was privileged to be supervised by Prof.\ Vincent Y.\ F.\ Tan and Prof.\ Wang Chi Cheung during her PhD study, and she worked with them as a research fellow between June 2021 and July 2022.

Dr.\ Zhong's research interests are in reinforcement learning, online machine learning and, in particular,  multi-armed bandits. Her work has been presented at top machine learning (ML) conferences including ICML and AISTATS, and also in top ML journals such as the {\em Journal of Machine Learning Research} (JMLR). She also serves as a reviewer for several conferences and journals including   AISTATS, ICLR, ICML, NeurIPS, TIT, TSP, and TMLR. She was recognized as a top reviewer of NeurIPS 2022.
\end{IEEEbiography}

\begin{IEEEbiography}[{\includegraphics[width=1in,height=1.25in,clip,keepaspectratio]{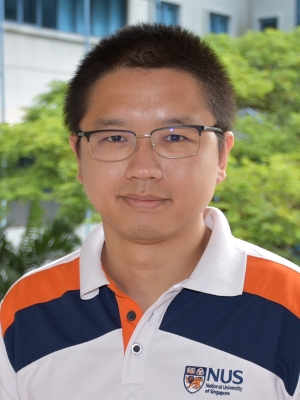}}]{Vincent Y.\ F.\ Tan} (S'07-M'11-SM'15)  was born in Singapore in 1981. He is currently an Associate Professor in the Department of Mathematics and the  Department of Electrical and Computer Engineering at the National University of Singapore (NUS). He received the B.A.\ and M.Eng.\ degrees in Electrical and Information Sciences from Cambridge University in 2005 and the Ph.D.\ degree in Electrical Engineering and Computer Science (EECS) from the Massachusetts Institute of Technology (MIT)  in 2011.  His research interests include information theory, machine learning, and statistical signal processing.

Dr.\ Tan received the MIT EECS Jin-Au Kong outstanding doctoral thesis prize in 2011, the NUS Young Investigator Award in 2014, the Singapore National Research Foundation (NRF) Fellowship (Class of 2018) and the NUS Young Researcher Award in 2019. He was also an IEEE Information Theory Society Distinguished Lecturer for 2018/9. He is currently serving as a Senior Area Editor of the {\em IEEE Transactions on Signal Processing} and an Associate Editor of Machine Learning for the {\em IEEE Transactions on Information Theory}. He is a member of the IEEE Information Theory Society Board of Governors.
\end{IEEEbiography}

\vfill

\end{document}